\documentclass[12pt,onecolumn,draftclsnofoot]{IEEEtran}

\usepackage{cite}
\usepackage{graphicx}

\usepackage{tikz}
\usetikzlibrary{decorations}
\usetikzlibrary{decorations.pathreplacing}
\usetikzlibrary{arrows,positioning,calc,fit,shapes.geometric}
\usetikzlibrary{arrows.meta}
\tikzset{>={Stealth[length=1.5mm, width=1.5mm]}}

\usepackage{xcolor}
\usepackage{amsfonts}
\usepackage{mathtools}
\newcommand{\mc}[1]{\mathcal{#1}}
\usepackage{mathrsfs}

\newcommand{\NN}{\mathbb{N}}

\newcommand{\PP}{\mathbb{P}}
\newcommand{\EE}{\mathbb{E}}

\usepackage{amsthm}
\newtheorem{theorem}{Theorem}
\newtheorem{proposition}{Proposition}

\newtheorem{remark}{Remark}
\newtheorem*{remark*}{Remark}


\DeclareMathOperator{\Marg}{Marg}

\usepackage[hidelinks]{hyperref}

\title{Single-Letterization of Rate-Cost Tradeoffs in Static Time-Invariant Control (STIC) Systems}
\author{Eray Unsal Atay, 
Venkat Chandrasekaran, 
Victoria Kostina 
\thanks{This work was supported in part by AFOSR grants \mbox{FA9550-23-1-0070} and FA9550-23-1-0204, by Caltech's Center for Sensing to Intelligence, and by the Carver Mead New Adventures Fund.}%
\thanks{The authors are with the California Institute of Technology, Pasadena, CA 91106 USA (email: eatay@caltech.edu; venkatc@caltech.edu; vkostina@caltech.edu).}%
}
\date{}

\begin{document}

\maketitle

\begin{abstract}
    In settings such as online recommendation platforms,
    medical decision-making with interim summaries,
    and stochastic control with static plants, the underlying system is often \emph{static}
    so that the
    next state depends only on the most recent action.
    Moreover,
    the true system state is high-dimensional and cannot be observed directly,
    so actions must be selected based on information received through a rate-limited communication link, which constrains the control performance.
    This leads to a fundamental rate-cost tradeoff---the minimum communication rate required to achieve a prescribed average control cost.
    We study this tradeoff for a class of static systems we term \emph{static time-invariant control} (STIC) systems, namely, control systems in which the next state depends on the history only through the most recent control action via a time-invariant kernel.
    For general control systems---not necessarily static or time-invariant---our prior work characterized this tradeoff through the limit of a sequence of directed-information minimization problems.
    However, the dimension of these problems grows with the time horizon, making the limiting formulation computationally intractable.
    In the present paper, we show for STIC systems that this growing-size formulation admits a single-letter characterization.
    This identifies STIC systems as the control-side counterpart of Noisy Output is the STate (NOST) channels---one of the few channel models whose feedback capacity admits a single-letter characterization.
    This yields, to our knowledge, the first constant-size convex characterization of the rate--cost tradeoff beyond linear quadratic Gaussian (LQG) control.
\end{abstract}


\section{Introduction} \label{sec:intro}

Systems that evolve based on the most recent decision or action
are ubiquitous in contemporary applications.
In online review platforms, the currently displayed rating or numerical score shapes subsequent user choices~\cite{Tang2019Bandit, Mansour2020Bayesian}.
In medical decision-making with public interim summaries,
the latest disclosed efficacy estimate influences later treatment and participation decisions by physicians and patients~\cite{Mansour2020Bayesian}.
Despite the ubiquity of such structure, systems with this property are referred to by different names in different domains:
in reinforcement learning, as a multi-armed bandit~\cite{slivkins2019introduction, lattimore2020bandit}
or a Markov decision process (MDP) with state-independent transitions~\cite{FloresHernandez2013Monotone, Kallenberg2013DermanLP};
in statistics, as a controlled i.i.d.\ process~\cite{Shimkin1993Extremal};
in game theory and multi-agent learning, as a stateless game~\cite{busoniu2010multiagent};
and in stochastic control, as a static (or memoryless) plant~\cite{kosut1998run, LipsaMartins2011, Blanchini2017ModelFreePlantTuning}.
While the terminology may vary, the underlying models share the same basic feature that the next state depends only on the most recent action;
formally, if $X_t$ denotes the system state and $U_t$ denotes the action at time $t$,
and we write $X_{[t]} \coloneq ( X_1,\ldots,X_t )$ and $U_{[t]} \coloneq ( U_1,\ldots,U_t )$,
then $X_{t+1}$ depends on the history $\left( X_{[t]} , U_{[t]} \right)$ only through $U_t$.
Equivalently,
\begin{equation} \label{eq:static}
    X_{t+1} \perp \left( X_{[t]} , U_{[t-1]} \right) \! \ \big| \ U_t .
\end{equation}
In this paper, we will call such systems \emph{static}.
Figure~\ref{fig:dep0} illustrates~\eqref{eq:static} by comparing
the state-dependence structure of a general (non-static) system with that of a static system.

In stochastic control, one typically
seeks to
choose a causal control policy
to minimize the
average per-stage cost~\cite{KostinaHassibi2019RateCostControl, KaraRaginskyYuksel2022Robustness}.
For a finite horizon $n$, the corresponding optimization problem can be written as
\begin{equation} \label{eq:classical-control}
    \inf_{ \left\{P_{U_t \mid X_{[t]}, U_{[t-1]}}\right\}_{t=1}^n } \frac{1}{n}\sum_{t=1}^n \EE[c(X_t,U_t)]
\end{equation}
where $c(\,\cdot\, , \cdot\,)$ is a nonnegative cost function.
When the controller has direct access to the state, the problem reduces to a classical fully observed stochastic control problem.
For Markovian models, this is the standard average-cost MDP, characterized by Bellman optimality equations and tractable in finite-state finite-action or linear-quadratic settings~\cite{puterman1994markov,Bertsekas2012DP,AndersonMoore1990LQMethods}.
For general nonlinear or continuous-state systems, however, the Bellman equation is infinite-dimensional and the problem remains computationally challenging~\cite{HernandezLermaLasserre1996MCP,ChowTsitsiklis1989DPComplexity}.  This classic formulation no longer applies directly when the communication from the system to the controller is rate-limited;
in particular, in such situations, the controller does not observe the state $X_t$ directly, but instead selects the action $U_t$ based on an encoded description of the state~\cite{Tatikonda2000, nair2007feedback}.
Consequently, the achievable control performance is constrained by the available communication rate, leading to a fundamental tradeoff between rate and cost~\cite{TatikondaSahaiMitter2004_StochasticLinearControl, SilvaDerpichOstergaardEncina2016_MinAvgRate, TanakaEsfahaniMitter2018_LQGMinDI, KostinaHassibi2019RateCostControl}.
This tradeoff is captured by the \emph{rate-cost function}~\cite{YukselBasar2013AgreementTeams, ZhaoChiaWeissman2014CompressionActions, KostinaHassibi2018, NakahiraXiaoKostinaDoyle2018BiomolecularControl, KostinaHassibi2019RateCostControl}, namely, the minimum encoding rate required to satisfy a prescribed average cost constraint.
The rate-cost function in rate-limited control is the analog of the rate-distortion function in source coding~\cite{Shannon1959, Berger1971}.
For general causal control systems---not necessarily static or time-invariant---the rate-cost function is given by the minimum directed information subject to control cost~\cite{TanakaEsfahaniMitter2018_LQGMinDI, KostinaHassibi2019RateCostControl}.
Specifically,
\begin{equation} \label{eq:F-D}
    F(D)
    \coloneq
    \lim_{n\to\infty}
    \min_{\substack{
        \left\{ P_{U_t \mid X_{[t]},U_{[t-1]}} \right\}_{t=1}^n \colon \\
        \frac{1}{n}\sum_{t=1}^n \EE[c(X_t,U_t)] \le D
    }}
    \frac{1}{n} I\!\left(X_{[n]} \to U_{[n]}\right),
\end{equation}
where $I \! \left( X_{[n]} \to U_{[n]} \right) \coloneq \sum_{t=1}^n I \! \left( X_{[t]} ; U_t \, \big| \, U_{[t-1]} \right)$ denotes the directed information~\cite{Massey1990} from the system state vector to the action vector.
Compared with the classical control formulation in~\eqref{eq:classical-control},
instead of minimizing the average control cost directly,
the optimization in~\eqref{eq:F-D} minimizes the information flow from the state process to the action process subject to an average-cost constraint.
Thus, for each target cost level $D$, $F(D)$ quantifies the minimum
directed information required to achieve that performance level.
Since $F(D)$ is characterized as the limit of a sequence of growing-size optimization problems, it is generally intractable to compute.

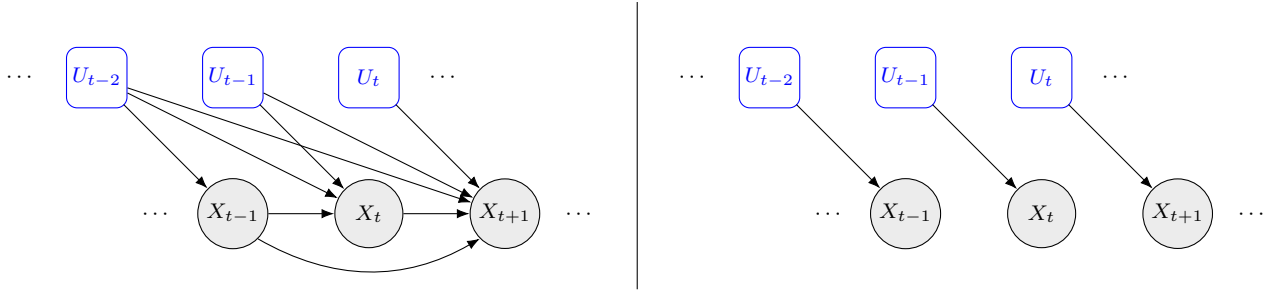
\begin{figure}[t]
    \centering


    \def\unodewidth{0.8cm}
    \def\unodeheight{0.8cm}
    \def\xnodesize{0.90cm}
    \def\nodeinnersep{2pt}
    \def\nodefont{\scriptsize}
    \def\unodefill{blue!0}
    \def\xnodefill{gray!15}

    \def\urow{0}
    \def\xrow{-1.8}
    \def\colsep{1.8}
    \def\dotgap{1.0}
    \def\leftx{0}
    \def\diagramgap{3.5}
    \def\separatorpad{1.0}
    \def\separatorwidth{0.4pt}
    \def\xbend{30}

    \def\diagsourceNWshift{0.05cm}

    \makebox[\textwidth][c]{%
    \hspace*{-0.1cm}%
    \begin{tikzpicture}[
        >=Latex,
        u node/.style={
            draw=blue,
            text=blue,
            rounded corners,
            align=center,
            minimum width=\unodewidth,
            minimum height=\unodeheight,
            inner sep=\nodeinnersep,
            fill=\unodefill,
            font=\nodefont
        },
        x node/.style={
            draw,
            circle,
            align=center,
            minimum size=\xnodesize,
            inner sep=\nodeinnersep,
            fill=\xnodefill,
            font=\nodefont
        },
        dot node/.style={
            font=\nodefont
        },
        arrow line/.style={
            -Latex
        }
    ]

        \pgfmathsetmacro{\rightx}{\leftx + 3*\colsep + \diagramgap}
        \pgfmathsetmacro{\separatorx}{\leftx + 3*\colsep + \diagramgap/2}

        \def\diagSource#1{([xshift=-\diagsourceNWshift,yshift=\diagsourceNWshift]#1.south east)}


        \node[u node] (LUtm2) at ({\leftx},\urow) {$U_{t-2}$};
        \node[u node] (LUtm1) at ({\leftx+\colsep},\urow) {$U_{t-1}$};
        \node[u node] (LUt)   at ({\leftx+2*\colsep},\urow) {$U_t$};

        \node[x node] (LXtm1) at ({\leftx+\colsep},\xrow) {$X_{t-1}$};
        \node[x node] (LXt)   at ({\leftx+2*\colsep},\xrow) {$X_t$};
        \node[x node] (LXtp1) at ({\leftx+3*\colsep},\xrow) {$X_{t+1}$};

        \node[dot node] at ({\leftx-\dotgap},\urow) {$\cdots$};
        \node[dot node] at ({\leftx+\colsep-\dotgap},\xrow) {$\cdots$};
        \node[dot node] at ({\leftx+2*\colsep+\dotgap},\urow) {$\cdots$};
        \node[dot node] at ({\leftx+3*\colsep+\dotgap},\xrow) {$\cdots$};

        \draw[arrow line] \diagSource{LUtm2} -- (LXtm1);
        \draw[arrow line] (LUtm2) -- (LXt);
        \draw[arrow line] (LUtm2) -- (LXtp1);

        \draw[arrow line] \diagSource{LUtm1} -- (LXt);
        \draw[arrow line] (LUtm1) -- (LXtp1);

        \draw[arrow line] \diagSource{LUt} -- (LXtp1);

        \draw[arrow line] (LXtm1) -- (LXt);
        \draw[arrow line] (LXt) -- (LXtp1);

        \draw[arrow line] (LXtm1.south east) to[bend right=\xbend] (LXtp1.south west);


        \draw[line width=\separatorwidth]
            ({\separatorx},{\urow+\separatorpad})
            --
            ({\separatorx},{\xrow-\separatorpad});


        \node[u node] (RUtm2) at ({\rightx},\urow) {$U_{t-2}$};
        \node[u node] (RUtm1) at ({\rightx+\colsep},\urow) {$U_{t-1}$};
        \node[u node] (RUt)   at ({\rightx+2*\colsep},\urow) {$U_t$};

        \node[x node] (RXtm1) at ({\rightx+\colsep},\xrow) {$X_{t-1}$};
        \node[x node] (RXt)   at ({\rightx+2*\colsep},\xrow) {$X_t$};
        \node[x node] (RXtp1) at ({\rightx+3*\colsep},\xrow) {$X_{t+1}$};

        \node[dot node] at ({\rightx-\dotgap},\urow) {$\cdots$};
        \node[dot node] at ({\rightx+\colsep-\dotgap},\xrow) {$\cdots$};
        \node[dot node] at ({\rightx+2*\colsep+\dotgap},\urow) {$\cdots$};
        \node[dot node] at ({\rightx+3*\colsep+\dotgap},\xrow) {$\cdots$};

        \draw[arrow line] \diagSource{RUtm2} -- (RXtm1);
        \draw[arrow line] \diagSource{RUtm1} -- (RXt);
        \draw[arrow line] \diagSource{RUt} -- (RXtp1);

    \end{tikzpicture}%
    }
    \caption{
        Graphs describing the conditional state distributions $\left\{ P_{X_t \mid X_{[t-1]}, U_{[t-1]}} \right\}_{t\ge1}$ in general (non-static) and static systems.
        Directed black edges indicate causal dependencies governing the system states.
        Left: in a general system, $X_{t+1}$ may depend on the entire history
        $(X_{[t]},U_{[t]})$.
        Right: in a static system, $X_{t+1}$ depends on the history only through the most recent action
        $U_t$, as in~\eqref{eq:static}.
        Equivalently, $P_{X_t \mid X_{[t-1]}, U_{[t-1]}} = P_{X_t \mid U_{t-1}}$.
    }
    \label{fig:dep0}
\end{figure}

In this paper, we study the rate-cost function of
\emph{static time-invariant control} (STIC) systems,
i.e., control systems
for which there exists a kernel $Q_{X \mid U'}$ such that
\begin{equation} \label{eq:static-2}
    P_{X_{t+1} \mid X_{[t]}, U_{[t]}} = P_{X_{t+1} \mid U_t} = Q_{X \mid U'} \quad \forall t\ge1 .
\end{equation}
The static dependence structure has already been captured by \eqref{eq:static}
and illustrated in Figure~\ref{fig:dep0}.
The key additional simplification in the STIC model in~\eqref{eq:static-2} is \emph{time-invariance}:
the probability law governing the dependence of the next state $X_{t+1}$ on the most recent action $U_t$ does not vary with time $t$.
We derive a single-letter characterization of the rate-cost function of
STIC systems.
This yields the first characterization of the rate-cost tradeoff as a single-letter convex optimization problem beyond linear quadratic Gaussian (LQG) control~\cite{KostinaHassibi2018, KostinaHassibi2019RateCostControl, TanakaEsfahaniMitter2018_LQGMinDI}.
The resulting single-letter characterization also reveals a structural duality\footnote{The notion of duality that we refer to throughout the paper should not be confused with duality in mathematical optimization.} between STIC rate-cost minimization and the feedback-capacity formula for NOST channels~\cite{chen2005capacity, NOST},
therefore identifying STIC systems as the control analog of NOST channels.

Our setting of rate-limited control of STIC systems departs in two essential ways from the prior literature on systems of the form~\eqref{eq:static-2} in the context of other sequential decision-making problems in reinforcement learning~\cite{slivkins2019introduction, lattimore2020bandit} or MDPs~\cite{puterman1994markov}.
The first key distinction is the information structure;
the observation mechanism of the controller is not fixed or perfect
(as in standard MDP formulations~\cite{puterman1994markov}) and the actions must be taken based on information conveyed through a communication link of limited rate.
As a result, the communication architecture
and how much information about the state is made available to the controller
become part of the design problem.
This makes the problem fundamentally information-theoretic as
the object of interest is the tradeoff between communication rate and achievable control performance.
The second essential distinction is that the system dynamics in~\eqref{eq:static-2} are treated as known,
which isolates the information-theoretic cost of communication from the statistical problem of learning the dynamics;
when the kernel is known, the question is not one of learning or exploration (as in reinforcement learning~\cite{sutton2018reinforcement}),
but of determining the minimum information that must be transmitted to attain a prescribed control cost.
This assumption is natural in model-based control settings such as model predictive control (MPC), where the plant model is typically specified from first principles or identified offline from data~\cite{Rawlings2000,RawlingsMayneDiehl2017MPC,VerheijenBreschiLazar2023}.
In contrast, the sequential decision-making literature typically focuses on performance criteria such as expected cumulative cost~\cite{derman1962sequential, DufourPiunovskiy2013ExpectedTotalCost}
or regret under a fixed observation model~\cite{slivkins2019introduction, lattimore2020bandit},
rather than on optimizing the observation channel together with the control action policy under a rate constraint.
Consequently, although STIC systems share a transition structure that also appears in standard sequential decision-making models, the rate-limited control problem shifts the central question from choosing actions under a fixed information pattern to determining how much information must be communicated to achieve a desired level of control performance.

\begin{remark}[Higher-order memory] \label{rem:higher-order-memory}
    The STIC framework also captures fixed finite action memory.
    Specifically, consider \emph{$k$-Markov static time-invariant control systems}, namely systems satisfying
    \begin{equation}
        P_{X_{t+1}\mid X_{[t]},U_{[t]}}
        =
        P_{X_{t+1}\mid U_t,\ldots,U_{t-k+1}}
        =
        Q_{X \mid \overline U'}
        \qquad \forall t\ge 1 .
    \end{equation}
    Define the augmented action
    \begin{equation}
        \widetilde U_t \coloneq (U_t,\ldots,U_{t-k+1})
        \in \mc U^k .
    \end{equation}
    Then the dynamics become
    \begin{equation}
        P_{X_{t+1}\mid X_{[t]},\widetilde U_{[t]}}
        =
        P_{X_{t+1}\mid \widetilde U_t},
    \end{equation}
    so the system is a STIC system of the form~\eqref{eq:static-2} over
    the augmented action alphabet $\mc U^k$ with kernel
    $Q_{X \mid \overline U'}$.
    The corresponding single-stage cost function is
    \begin{equation}
        \widetilde c(x,\widetilde u)
        \coloneq
        c(x,u_1),
        \qquad
        \widetilde u=(u_1,\ldots,u_k).
    \end{equation}
\end{remark}

\paragraph*{Notation}
Random variables are denoted by uppercase letters (e.g.,~$X,Y$), and their realizations by lowercase letters (e.g.,~$x,y$).
Calligraphic letters denote the alphabets the random variables lie in (e.g.,~$\mc X,\mc Y$).
For random variables $X$ and $Y$, $P_{X,Y}$ denotes the joint distribution of $(X,Y)$,
and $P_{X \mid Y=y}$ the conditional distribution of $X$ given $Y=y$.
Independence is denoted by $X\perp Y$.
All information quantities are measured in bits,
and all logarithms are base~$2$.
The length of a binary string $w \in \{0,1\}^\star$ is denoted by $\ell(w) \in \NN$.

\subsection{Formal problem setup}

In
the stochastic control system
depicted in Figure~\ref{fig:sys1},
at each time $t$, the controller does not observe the system state $X_t \in \mc X$ directly but instead receives an encoded finite binary string $B_t \in \mc B_t \subset \{0,1\}^\star$ from a prefix-free codebook $\mc B_t$ of binary strings.
Based on this, the controller chooses a control action $U_t \in \mc U$, which incurs a per-stage control cost $c(X_t,U_t)$.
We assume that the state and action alphabets $\mc X, \mc U$ are finite.

The state process is governed by a given sequence of known conditional distributions\linebreak
$\bigl\{ P_{X_t \mid X_{[t-1]},U_{[t-1]}} \bigr\}_{t=1}^n$.
The communication--control architecture is then described by two sequences of design kernels:
\begin{itemize}
    \item the encoder kernels
    $\left\{ P_{B_t \mid X_{[t]},B_{[t-1]}} \right\}_{t=1}^n$,
    where each kernel generates a message $B_t$ based on the observed states $X_{[t]}$ and the past encoded messages $B_{[t-1]}$;
    
    \item the control kernels
    $\left\{ P_{U_t \mid B_{[t]},U_{[t-1]}} \right\}_{t=1}^n$,
    where each kernel selects the control action using the received messages together with the past control inputs.
\end{itemize}
Combined with the system dynamics $\left\{ P_{X_t \mid X_{[t-1]},U_{[t-1]}} \right\}_{t=1}^n$,
these design kernels determine the joint distribution of $\left( X_{[n]}, B_{[n]}, U_{[n]} \right)$.

\begin{figure}[!t]
    \centering

    \begin{tikzpicture}[
        >=Latex,
        node distance=15mm and 15mm,
        block/.style={
            draw, rounded corners,
            align=center,
            minimum height=10mm, inner sep=4pt, fill=white
        },
        plant/.style={block, fill=gray!15, minimum width=2.0cm, minimum height=1.2cm},
        enc/.style={block, minimum width=2.0cm, minimum height=1.2cm},
        ctrl/.style={block, minimum width=2.0cm, minimum height=1.2cm},
        line/.style={-Latex}
    ]

    \node[plant] (source) {System};
    \node[enc, right=of source] (encoder) {Encoder};
    \node[ctrl, right=of encoder] (controller) {Controller};

    \draw[line] (source.east) -- node[above] {$X_t$} (encoder.west);
    \draw[line] (encoder.east) -- node[above] {$B_t$} (controller.west);
    \draw[line]
          (controller.south) -- ++(0,-0.8) coordinate (ctrlbottom)
          -- (source |- ctrlbottom) coordinate[midway] (mid)
          -- (source.south);

    \node[below] at (mid) {$U_t$};
    \node[below=1.5em] at (mid) {$\text{per-stage control cost } c(X_t,U_t)$};

    \end{tikzpicture}
    
    \caption{Discrete-time control system with a rate-limited communication link to the controller.
    The shaded system block is fixed and known; the white encoder and controller blocks are design components.}
    \label{fig:sys1}
\end{figure}
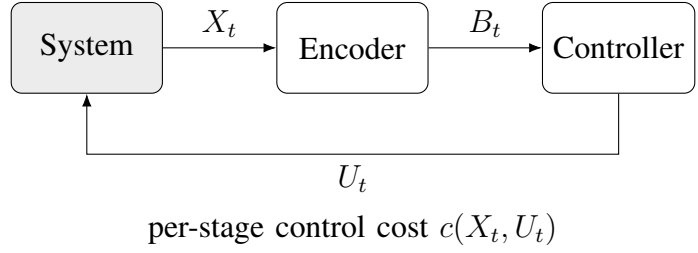

We are interested in the minimum average rate
\begin{equation} \label{eq:rate}
    \frac{1}{n}\sum_{t=1}^n \EE[\ell(B_t)]
\end{equation}
that makes
it possible to keep the average control cost~\eqref{eq:classical-control} below a desired threshold $D\ge0$
asymptotically as the time horizon $n \to \infty$.
There is a tradeoff between the encoding rate~\eqref{eq:rate} and
the minimal achievable control cost~\eqref{eq:classical-control}~\cite{TatikondaSahaiMitter2004_StochasticLinearControl, SilvaDerpichOstergaardEncina2016_MinAvgRate, TanakaEsfahaniMitter2018_LQGMinDI, KostinaHassibi2019RateCostControl},
called the fundamental \emph{rate-cost tradeoff} in rate-limited control.
This tradeoff is studied in~\cite{atay2026ratecost} via the information-theoretic \emph{rate-cost function}
$F(D)$ given in~\eqref{eq:F-D}.
$F(D)$ is a lower bound on the minimum achievable asymptotic average prefix-free coding rate under the cost constraint $D$~\cite{atay2026ratecost}.
Moreover,~\cite{atay2026ratecost} proves a matching achievability result up to additive logarithmic terms, so the operational minimum average rate and $F(D)$ differ by at most $O \bigl( \log F(D) \bigr)$.
Thus,
although not an exact operational characterization,
$F(D)$ is an information-theoretic benchmark that governs the rate-cost tradeoff.

Each minimization instance in the limit in~\eqref{eq:F-D} is convex~\cite{Charalambous2016DirectedInformation, atay2026ratecost},
but this alone does not make the formulation tractable.
Indeed, for horizon $n$, the design variables are the causal kernels
$P_{U_t \mid X_{[t]},U_{[t-1]}}$, $t=1,\ldots,n$,
whose number of free parameters grows on the order of
$\sum_{t=1}^n (|\mc U|-1)|\mc X|^t|\mc U|^{t-1}$,
and hence exponentially in the horizon $n$.
Moreover, $F(D)$ is not the value of any of the finite-dimensional convex programs, but the limit of the normalized minimum values.
Thus, solving the finite-horizon problems numerically does not yield a characterization of the infinite-horizon tradeoff, and
a constant-size characterization is important both computationally and conceptually.   In particular, such a description makes the fundamental limit $F(D)$ accessible through a single optimization problem whose dimension depends only on the alphabets rather than on the horizon.  Moreover, it identifies the stationary one-step object governing the optimal rate-cost tradeoff, makes the structure of optimal solutions more transparent, yields qualitative properties of the rate-cost function, and guides the design of encoder--controller architectures.
It is therefore of interest to identify classes of systems---specified by the state transition kernels $\left\{ P_{X_t \mid X_{[t-1]},U_{[t-1]}} \right\}_{t=1}^n$---for which $F(D)$ admits a constant-size characterization.

\subsection{Contributions}

For STIC systems defined in~\eqref{eq:static-2},
we establish the following single-letter characterization of the rate-cost function given by the growing-size limit in~\eqref{eq:F-D}
via symmetry reduction.

\begin{theorem} \label{thm:single-letter}
    The rate-cost function $F(D)$ of a STIC system
    is given by
    \begin{equation}
        \begin{aligned} \label{eq:single-letter}
            F(D)
            = & \min_{P_{U,X,U'}} I\!\left(X;U\mid U'\right) \\[-1ex]
            & \quad \textup{s.t.} \quad
            \begin{aligned}[t]
                \EE\!\left[c(X,U)\right] &\le D \, , \\[-2ex]
                P_U &= P_{U'} \, , \\[-2ex]
                P_{X\mid U'} &= Q_{X\mid U'} \, .
            \end{aligned}
        \end{aligned}
    \end{equation}
\end{theorem}

\noindent
Theorem~\ref{thm:single-letter} shows that, for STIC systems, the causal rate-cost function admits a convex single-letter characterization in terms of the minimal conditional mutual information between the system state and the control action, given the previous action.
This yields, to our knowledge, the first single-letter convex optimization characterization of the rate-cost tradeoff for a class of control systems beyond LQG control~\cite{KostinaHassibi2018, KostinaHassibi2019RateCostControl, TanakaEsfahaniMitter2018_LQGMinDI}.

Our proof methodology is as follows.
We leverage time-invariance (stationarity) to obtain a sequence of convex programs that are invariant under cyclic shifts of the time index, i.e., \emph{symmetric}.
We then build on the symmetry reduction argument~\cite[Ex.~4.4]{Boyd2004convex}
in the spirit of the methodology developed in~\cite{Atay2026PosetMarkov} for obtaining constant-size reductions---there, for the feedback capacity of NOST channels (see the subsequent discussion for more on the relation between the NOST channels and STIC systems).
This yields a reduced convex optimization problem for each time horizon;
the single-letter formula for $F(D)$ is then obtained in the limit as the horizon tends to infinity.
We emphasize that our method does not turn each finite-horizon problem into a constant-size one.
Rather, symmetry reduction applies in the asymptotic limit and shows that the growing-size sequence defining $F(D)$ has an equivalent single-letter characterization.
A notable methodological point is that the symmetry-reduction approach we developed in~\cite[Sec.~2]{Atay2026PosetMarkov} was applied to NOST channels in an unconstrained input setting since the setting did not involve an input-cost or power constraint,
whereas the present rate-cost problem includes the average-cost constraint $\EE[c(X,U)] \le D$.
Thus, the present work shows that the symmetry-reduction methodology extends beyond unconstrained channel-capacity problems to broader input-constrained information-theoretic problems.

The single-letter rate-cost characterization of STIC systems in Theorem~\ref{thm:single-letter} has a close counterpart on the channel-coding side, namely the feedback capacity of the NOST channel model studied in~\cite{chen2005capacity, NOST} and in~\cite[Sec.~2]{Atay2026PosetMarkov}.
NOST channels have one-step memory and are specified by a fixed conditional law $Q_{Y\mid X,Y'}$.
The NOST feedback capacity can be written as the concave maximization problem
\begin{equation} \label{eq:nost-c}
        \begin{aligned}
            C^\textup{fb}_{\mathrm{NOST}}
            \, = \,
            & \max_{P_{X , Y'}}
            I \! \left( X;Y \mid Y' \right) \\[-1ex]
            & \ \ \textup{s.t.} \ \ P_Y = P_{Y'}
    \end{aligned}
\end{equation}
where concavity follows from the fact that, for fixed $P_{Y\mid X,Y'}$, the functional $I \! \left( X;Y \mid Y' \right)$ is concave in the joint distribution $P_{X, Y'}$.
In Table~\ref{table:nost-rcf},
we highlight a form of \emph{duality} between
the single-letter characterization of the feedback capacity of NOST channels in~\cite[Sec.~2]{Atay2026PosetMarkov} and that of the rate-cost function of STIC systems characterized by~\eqref{eq:static-2}.
This provides additional insight into the notion of duality between the source and channel coding problems that was studied previously for memoryless i.i.d.\ sources and channels~\cite{pradhan2003duality,cover2002duality}.

\begin{table}[ht]
    \caption{Duality between NOST channel feedback capacity and STIC system rate-cost function.}
    \label{table:nost-rcf}
    \centering
    \renewcommand{\arraystretch}{1.5}
    \begin{tabular}{c|c|c}
                                    & NOST channel feedback capacity                     & STIC system rate-cost function                         \\ \cline{1-3}
        Objective function          & $I(X ; Y \mid Y^\prime)$                  & $I(X ; U \mid U^\prime)$                  \\ \cline{1-3}
        Fixed kernel              & {\color{red} $P_{Y \mid X , Y^\prime}$}   & {\color{blue} $P_{X \mid U^\prime}$}      \\ \cline{1-3}
        Input kernel           & {\color{blue} $P_{X \mid Y^\prime}$}      & {\color{red} $P_{U \mid X , U^\prime}$}  \\ \cline{1-3}
        Constraint                  & $P_{Y} = P_{Y^\prime}$                    & $P_{U} = P_{U^\prime}$                    \\ \cline{1-3}
        Type of optimization        & Concave maximization                      & Convex minimization
    \end{tabular}
\end{table}

In the NOST channel setting, the channel law $P_{Y\mid X,Y'}$ is fixed and one designs the input distribution $P_{X\mid Y'}$ to maximize the information flow $I \! \left( X;Y \mid Y' \right)$, subject to the stationarity constraint $P_Y=P_{Y'}$.
In our control setting of STIC systems, the state kernel $P_{X\mid U'} = Q_{X\mid U'}$ is fixed and one designs the control strategy $P_{U\mid X,U'}$ to minimize $I \! \left( X;U\mid U' \right)$, while meeting the cost constraint $\EE[c(X,U)]\le D$ and the stationarity constraint $P_U=P_{U'}$.
Thus, the two formulations mirror each other: the fixed kernel and the optimized kernel swap roles, and concave maximization in the channel problem becomes convex minimization in the control problem, while the same cyclic-shift invariance leads to the same type of constant-size reduction via convexity.
This viewpoint identifies STIC systems as the control analog of NOST channels.

\subsection{Related work}

\paragraph*{(Sequential) source coding}
In source coding, the rate-distortion function characterizes the minimum rate needed to represent a source within a prescribed average distortion. For a general source, this quantity is given by a limit of growing-size convex programs as
\begin{equation} \label{eq:rd-classical}
    F(D) = \lim_{n\to\infty} \frac1n \inf_{\substack{ P_{\hat X_{[n]} \mid X_{[n]}} \colon \\
    \frac{1}{n} \sum_{t=1}^n \EE\left[d\left(X_{t}, \hat X_{t}\right)\right] \le D }} I \! \left( X_{[n]} ; \hat X_{[n]} \right)
\end{equation}
in the classical block coding setting~\cite{Berger1971, CoverThomas2006}.
Single-letter
characterizations of the limit in~\eqref{eq:rd-classical} are known for memoryless i.i.d.\ sources~\cite{Shannon1959, Berger1968},
and for several Gaussian families under mean-square distortion,
including stationary Gaussian sources via the Shannon--Kolmogorov reverse-water-filling formula~\cite{Kolmogorov1956},
the discrete-time Wiener process~\cite{Berger1970Wiener},
and certain nonstationary Gaussian autoregressive sources~\cite{Gray1970AR,HashimotoArimoto1980}.
The sequential setting differs from the classical one in that each $\hat X_t$ must be produced before future source symbols are revealed, which models zero-delay decoding~\cite{NeuhoffGilbert1982, Yang2011}.
The relevant information measure becomes directed information~\cite{Massey1990},
and the associated \emph{sequential (causal)} rate-distortion function is given by~\cite{Tatikonda2000}
\begin{equation} \label{eq:rd-causal}
    F^{c}(D) = \lim_{n\to\infty} \frac1n \inf_{\substack{ \left\{ P_{\hat X_{t} \mid X_{[t]}, \hat X_{[t-1]}} \right\}_{t=1}^n \colon \\
    \frac{1}{n} \sum_{t=1}^n \EE\left[d\left(X_{t}, \hat X_{t}\right)\right] \le D}} I \! \left( X_{[n]} \to \hat X_{[n]} \right) .
\end{equation}
This form of the directed-information minimization was suggested by earlier finite-horizon results:
\cite[Thm.~1]{ViswanathanBerger2000} established the two-stage case, and~\cite[Thm.~1]{Yang2011} extended it to an arbitrary horizon $n$.
A more general model with $K$ observers with indirect access to the source was studied in~\cite{kostina2021ceo};
the special case $K=1$ with noiseless observations reduces to sequential source coding.
Closed-form or otherwise tractable characterizations of the limit in~\eqref{eq:rd-causal} are known only in special cases:
for scalar Gauss--Markov sources under quadratic distortion~\cite{KostinaHassibi2018};
for multidimensional Gauss--Markov sources via semidefinite programming~\cite{Tanaka},
and for a restricted class of stationary Markov sources with decoder feedforward~\cite{VenkataramananPradhan2010}.
These works study coding of sources that are not controlled, and therefore do not address control systems.

\paragraph*{Rate-limited control}
Control under communication constraints studies how much information sensors and controllers must exchange to achieve a desired control performance.
From an information-theoretic perspective, the control cost can be interpreted as a distortion measure, which leads to the rate-cost function~\cite{KostinaHassibi2018}.
This function denotes the minimum communication rate required to ensure that the control cost stays below a given threshold, thereby allowing rate-distortion theory to be applied to control problems~\cite{KostinaHassibi2018, KostinaHassibi2019RateCostControl, NakahiraXiaoKostinaDoyle2018BiomolecularControl}.
In the special case of LQG control, the rate-cost tradeoff is governed by $F(D)$ defined in~\eqref{eq:F-D}~\cite{KostinaHassibi2018, KostinaHassibi2019RateCostControl, TanakaEsfahaniMitter2018_LQGMinDI}. Moreover, for LQG systems, this quantity admits a more explicit characterization than~\eqref{eq:F-D}.
It is shown in~\cite{TanakaEsfahaniMitter2018_LQGMinDI} that among all control policies satisfying a given LQG cost constraint, the minimum directed information is achieved by a linear Gaussian policy that consists of a virtual sensor, a Kalman filter, and a certainty-equivalent controller. Furthermore, the asymptotic quantity $F(D)$ defined in~\eqref{eq:F-D} admits a single-letter semidefinite programming characterization,
and a closed-form expression for this semidefinite program in the scalar LQG case is given in~\cite{KostinaHassibi2019RateCostControl}.
Outside the LQG setting,~\cite{atay2026ratecost} demonstrated that the directed-information expression in~\eqref{eq:F-D} remains the fundamental quantity governing the rate-cost tradeoff in general nonlinear control systems, thereby identifying $F(D)$ in~\eqref{eq:F-D} as the rate-cost function in a broad sense. However, because~\eqref{eq:F-D} is defined through a growing-horizon limit, it is generally not amenable to a constant-size reduction. Prior to the present work, no class of control systems beyond LQG control was known to admit such a characterization.

\paragraph*{Static systems}
In control, system models satisfying~\eqref{eq:static-2}
are typically referred to as static (or memoryless) plants~\cite{kosut1998run, LipsaMartins2011, Blanchini2017ModelFreePlantTuning}.
Settings in which the next outcome or state is determined by the most recent action have also been studied outside the control literature.
Arguably the most prominent example is the multi-armed bandit problem~\cite{lattimore2020bandit, slivkins2019introduction}, whose modern treatment traces back to Robbins's formulation of sequential allocation under uncertainty~\cite{Robbins1952}.
In this literature, the counterpart of control cost is known as \emph{regret},
which is the gap between the reward earned by the policy and that of
the optimal decision.
Classical policies with regret guarantees include $\varepsilon$-greedy~\cite{auer2002finite, sutton2018reinforcement}, upper confidence bound (UCB)~\cite{auer2002finite, lai1985asymptotically}, and Thompson sampling~\cite{thompson1933likelihood, agrawal2012analysis, russo2018tutorial}.
The central difference between the multi-armed bandit setting and our control setting is that here the kernel~\eqref{eq:static-2} is assumed known and the focus is on the induced rate-cost tradeoff, whereas in the bandit setting the kernel is unknown and the goal is to minimize regret with respect to the optimal policy while learning it.

\paragraph*{Single-letter reductions}
Reducing growing-size optimization problems to constant-size ones is important in several areas, including extremal combinatorics~\cite{raymond2018symmetric,raymond2018symmetry,lovasz2012large,brosch_thesis}, quantum information~\cite{huber2021positive,klep2018positive,klep2022optimization}, and mean-field game theory~\cite{lasry2007mean}. In these settings, the underlying optimizations are group-invariant, and symmetry reduction often yields constant-size formulations~\cite{raymond2018symmetric,raymond2018symmetry,riener2013exploiting,levin2023free}.
In information theory,
this approach was used in~\cite[Sec.~2]{Atay2026PosetMarkov} to recover the convex single-letter formula for the feedback capacity of the NOST channel, originally proved in~\cite{NOST}.
The present work is the first to apply symmetry reduction to a class of control systems, yielding a constant-size characterization of the rate-cost function of STIC systems.
In the NOST channel setting~\cite[Sec.~2]{Atay2026PosetMarkov}, the optimization over the input strategy did not involve an additional input-cost or power constraint. The present work extends the utility of symmetry reduction by demonstrating that it can be used to single-letterize problems with input constraints.

\paragraph*{Channel--source coding duality}
Prior work has studied a duality between channel coding and lossy source coding by identifying channel capacity and rate-distortion formulas that become mirror images after an exchange of variables and constraints~\cite{pradhan2003duality,cover2002duality}. In the classical point-to-point setting, channel capacity fixes the channel law $P_{Y\mid X}$ and maximizes $I(X;Y)$ over the input law $P_X$, whereas rate-distortion fixes the source law $P_X$ and minimizes $I \! \left( X ; \hat X \right)$ over the reconstruction kernel $P_{\hat X\mid X}$ subject to a distortion constraint; thus the fixed and optimized kernels interchange, and maximization becomes minimization~\cite{pradhan2003duality,cover2002duality}.
This viewpoint was formalized
as a functional duality and was extended to coding problems with side information~\cite{pradhan2003duality}.
A unified treatment of channel capacity and rate-distortion with two-sided state information
was given in~\cite{cover2002duality},
showing that several classical formulas fit the same dual template.
Our comparison between the NOST channel capacity problem~\cite{NOST, Atay2026PosetMarkov} and the STIC system rate-cost problem
is in this same spirit.

\bigskip
The remainder of the paper is organized as follows.
In Section~\ref{sec:static}, we present the main steps yielding Theorem~\ref{thm:single-letter}.
In Section~\ref{sec:proofs},
we give the proof of Theorem~\ref{thm:single-letter} and full proofs of the intermediate propositions.
In Section~\ref{sec:numer}, we provide a numerical illustration of the single-letter characterization in Theorem~\ref{thm:single-letter}.
In Section~\ref{sec:dual}, we further elaborate on the duality between STIC systems and NOST channels by identifying dual families of kernels.

\section{Main Steps in Proof of Theorem~\ref{thm:single-letter}} \label{sec:static}

We prove Theorem~\ref{thm:single-letter} through a sequence of simplifications of the multi-letter rate-cost formulation in~\eqref{eq:F-D}.
Propositions~\ref{prop:simplification}--\ref{prop:cyclify} reduce the original growing-horizon problem to an equivalent cyclically symmetric formulation in the infinite-horizon limit.
Applying symmetry reduction to this formulation yields the single-letter convex characterization in~\eqref{eq:single-letter}.

For a general control system with arbitrary state kernels $\left\{ P_{X_t \mid X_{[t-1]},U_{[t-1]}} \right\}_{t=1}^n$,
an optimal action policy $\left\{P_{U_t \mid X_{[t]}, U_{[t-1]}}\right\}_{t=1}^n$ may have each $U_t$ depend on the full history $\left(X_{[t]},U_{[t-1]}\right)$, as illustrated in Figure~\ref{fig:dep1}.
Under additional Markov structure, however, this dependence can be reduced.
In particular, for STIC systems, the actions $U_t$ can---without loss of optimality---be restricted to depend only on the current state $X_t$ and the previous action $U_{t-1}$, rather than on the entire history $\left(X_{[t]},U_{[t-1]}\right)$.

\begin{figure}[t]
    \centering

    \begin{tikzpicture}[
            >=Latex,
            node distance=18mm and 22mm,
            block/.style={
                draw, rounded corners,
                align=center,
                minimum height=10mm, inner sep=4pt, fill=white
            },
            state/.style={block, minimum width=1.2cm, minimum height=1.2cm},
            ustate/.style={state, draw=blue, text=blue, fill=blue!0},
            line/.style={-Latex, blue}
        ]

        \node[ustate] (Utm1) {$U_{t-1}$};
        \node[ustate, right=of Utm1] (Ut) {$U_t$};
        \node[ustate, right=of Ut] (Utp1) {$U_{t+1}$};

        \node[state, circle, fill=gray!15, below=of Utm1] (Xtm1) {$X_{t-1}$};
        \node[state, circle, fill=gray!15, below=of Ut] (Xt) {$X_t$};
        \node[state, circle, fill=gray!15, below=of Utp1] (Xtp1) {$X_{t+1}$};

        \node[left=8mm of Utm1] {$\cdots$};
        \node[left=8mm of Xtm1] {$\cdots$};

        \node[right=8mm of Utp1] {$\cdots$};
        \node[right=8mm of Xtp1] {$\cdots$};

        \draw[line] (Xtm1) -- (Utm1);
        \draw[line] (Xtm1) -- (Ut);
        \draw[line] (Xtm1) -- (Utp1);

        \draw[line] (Xt) -- (Ut);
        \draw[line] (Xt) -- (Utp1);

        \draw[line] (Xtp1) -- (Utp1);

        \draw[line] (Utm1) -- (Ut);
        \draw[line] (Ut) -- (Utp1);

        \draw[line] (Utm1.north east) to[bend left=30] (Utp1.north west);
    \end{tikzpicture}

    \caption{
        Graph describing the conditional action distributions $\left\{P_{U_t \mid X_{[t]}, U_{[t-1]}}\right\}_{t\ge1}$ under a general causal action policy.
        Directed blue edges indicate causal dependencies governing the control actions.
        Each action $U_t$ may depend on the full history $\left(X_{[t]},U_{[t-1]}\right)$.
    }
    \label{fig:dep1}
\end{figure}
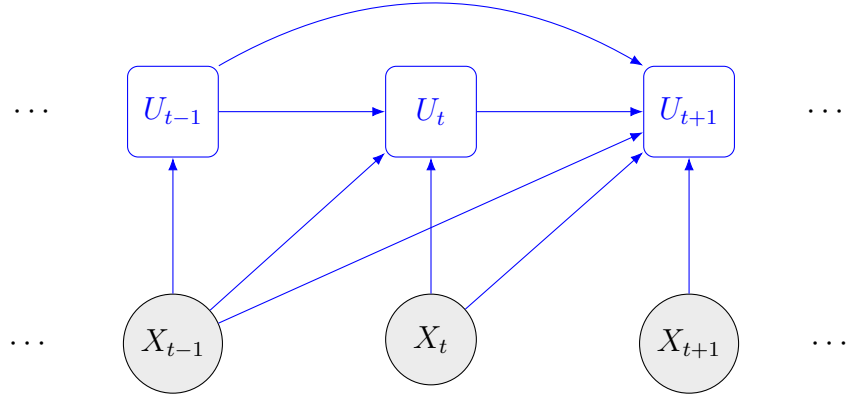

\begin{proposition}[History simplification] \label{prop:simplification}
    Each instance of minimization in~\eqref{eq:F-D} can be simplified as
    \begin{equation} \label{eq:simplification}
        \min_{\substack{ \left\{P_{U_t \mid X_{[t]}, U_{[t-1]}}\right\}_{t=1}^n \colon \\ \frac{1}{n} \sum_{t=1}^n \EE\left[c\left(X_t, U_t\right)\right] \le D}} \sum_{t=1}^n I \! \left( X_{[t]} ; U_t \mid U_{[t-1]} \right)
        =
        \min_{\substack{ \left\{P_{U_t \mid X_t, U_{t-1}}\right\}_{t=1}^n \colon \\ \frac{1}{n} \sum_{t=1}^n \EE\left[c\left(X_t, U_t\right)\right] \le D}} \sum_{t=1}^n I(X_t ; U_t \mid U_{t-1}),
    \end{equation}
    where the first summand on both sides is $I(X_1;U_1)$.
\end{proposition}

\begin{proof}[Proof sketch]
    We show that for any feasible history-dependent control law $\left\{ P_{U_t \mid X_{[t]},U_{[t-1]}} \right\}_{t=1}^n$,
    replacing it by the induced $\left\{ P_{U_t \mid X_t,U_{t-1}} \right\}_{t=1}^n$ causes no loss of optimality,
    proving the direction $\ge$.
    Conversely, any kernel of the form $P_{U_t \mid X_t,U_{t-1}}$ is feasible for the original problem.
    We show that under this choice, we have the equality
    \begin{equation}
        I \! \left( X_{[t]} ; U_t \, \big| \, U_{[t-1]} \right) = I(X_t ; U_t \mid U_{t-1}) \quad \forall t ,
    \end{equation}
    which proves the direction $\le$.
    The full proof is given in Section~\ref{subsec:proof-prop:simplification}.
\end{proof}

Proposition~\ref{prop:simplification} shows that for STIC systems, the actions designed to minimize directed information can be restricted to first-order Markov policies without loss of optimality.
The resulting first-order policy structure is shown in Figure~\ref{fig:dep1.5}: each action $U_t$ uses only the current state $X_t$ and the immediately preceding action $U_{t-1}$.
Combining the system state dependency in Figure~\ref{fig:dep0}
and the control action dependency in Figure~\ref{fig:dep1.5}
gives the overall dependency structure of STIC systems after history simplification,
depicted in Figure~\ref{fig:dep2}.

\begin{figure}[b]
    \centering

    \begin{tikzpicture}[
            >=Latex,
            node distance=18mm and 22mm,
            block/.style={
                draw, rounded corners,
                align=center,
                minimum height=10mm, inner sep=4pt, fill=white
            },
            state/.style={block, minimum width=1.2cm, minimum height=1.2cm},
            ustate/.style={state, draw=blue, text=blue, fill=blue!0},
            line/.style={-Latex, blue}
        ]

        \node[ustate] (Utm1) {$U_{t-1}$};
        \node[ustate, right=of Utm1] (Ut) {$U_t$};
        \node[ustate, right=of Ut] (Utp1) {$U_{t+1}$};

        \node[state, circle, fill=gray!15, below=of Utm1] (Xtm1) {$X_{t-1}$};
        \node[state, circle, fill=gray!15, below=of Ut] (Xt) {$X_t$};
        \node[state, circle, fill=gray!15, below=of Utp1] (Xtp1) {$X_{t+1}$};

        \node[left=8mm of Utm1] {$\cdots$};
        \node[left=8mm of Xtm1] {$\cdots$};

        \node[right=8mm of Utp1] {$\cdots$};
        \node[right=8mm of Xtp1] {$\cdots$};

        \draw[line] (Xtm1) -- (Utm1);
        \draw[line] (Xt) -- (Ut);
        \draw[line] (Xtp1) -- (Utp1);

        \draw[line] (Utm1) -- (Ut);
        \draw[line] (Ut) -- (Utp1);

    \end{tikzpicture}

    \caption{
        Graph describing the conditional action distributions $\left\{P_{U_t \mid X_{t}, U_{t-1}}\right\}_{t\ge1}$ under a history-simplified causal action policy.
        Each action $U_t$ depends on the history only through $(X_t,U_{t-1})$.
    }
    \label{fig:dep1.5}
\end{figure}
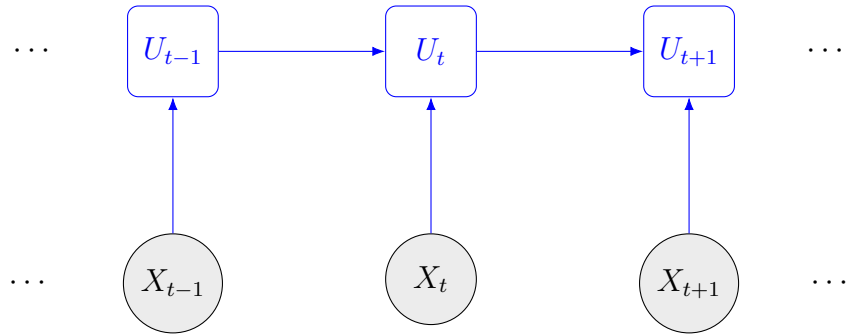

Combining Proposition~\ref{prop:simplification} with the definition of $F(D)$ in~\eqref{eq:F-D} yields the simplified expression
\begin{equation} \label{eq:simplified}
    F(D) =
    \lim_{n \to \infty}
    \min_{\substack{ \left\{P_{U_t \mid X_t, U_{t-1}}\right\}_{t=1}^n \colon \\ \frac{1}{n} \sum_{t=1}^n \EE\left[c\left(X_t, U_t\right)\right] \le D}}
    \dfrac{1}{n} \sum_{t=1}^n I(X_t ; U_t \mid U_{t-1}) .
\end{equation}

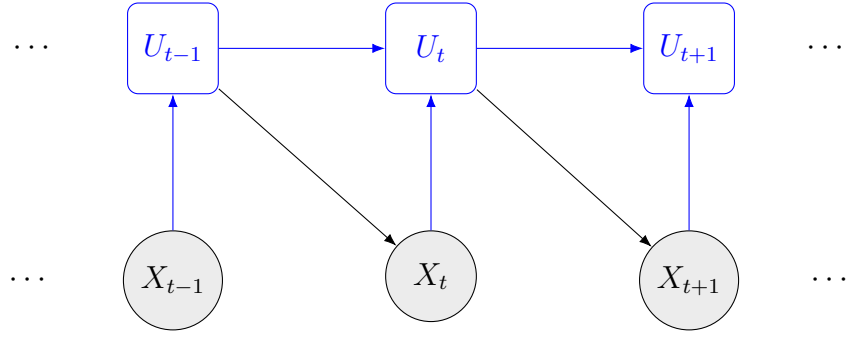
\begin{figure}[t]
    \centering

    \begin{tikzpicture}[
            >=Latex,
            node distance=18mm and 22mm,
            block/.style={
                draw, rounded corners,
                align=center,
                minimum height=10mm, inner sep=4pt, fill=white
            },
            state/.style={block, minimum width=1.2cm, minimum height=1.2cm},
            ustate/.style={state, draw=blue, text=blue, fill=blue!0},
            line/.style={-Latex},
            uline/.style={-Latex, blue}
        ]

        \node[ustate] (Utm1) {$U_{t-1}$};
        \node[ustate, right=of Utm1] (Ut) {$U_t$};
        \node[ustate, right=of Ut] (Utp1) {$U_{t+1}$};

        \node[state, circle, fill=gray!15, below=of Utm1] (Xtm1) {$X_{t-1}$};
        \node[state, circle, fill=gray!15, below=of Ut] (Xt) {$X_t$};
        \node[state, circle, fill=gray!15, below=of Utp1] (Xtp1) {$X_{t+1}$};

        \node[left=8mm of Utm1] {$\cdots$};
        \node[left=8mm of Xtm1] {$\cdots$};

        \node[right=8mm of Utp1] {$\cdots$};
        \node[right=8mm of Xtp1] {$\cdots$};

        \draw[uline] (Xtm1) -- (Utm1);
        \draw[uline] (Xt) -- (Ut);
        \draw[uline] (Xtp1) -- (Utp1);

        \draw[uline] (Utm1) -- (Ut);
        \draw[uline] (Ut) -- (Utp1);

        \draw[line] (Utm1) -- (Xt);
        \draw[line] (Ut) -- (Xtp1);
    \end{tikzpicture}

    \caption{
        Graph describing the conditional state distributions
        $\left\{ P_{X_t \mid U_{t-1}} \right\}_{t\ge1}$
        and conditional action distributions
        $\left\{P_{U_t \mid X_{t}, U_{t-1}}\right\}_{t\ge1}$
        for a STIC system after history simplification.
        Black and blue edges indicate dependencies governing the system states and control actions, respectively.
        Each system state $X_{t+1}$ depends on the history only through $U_t$;
        each action $U_t$ depends on the history only through $(X_t,U_{t-1})$.
    }
    \label{fig:dep2}
\end{figure}

The next step is to make the dependence on the initial distribution explicit.
Since the system is time-invariant and governed by the fixed kernel $Q_{X \mid U'}$,
it is natural to anticipate that
the long-term behavior of the controlled process would be insensitive to the choice of $P_{X_1}$.
The following proposition shows that, in the infinite-horizon limit, we may in fact minimize over $P_{X_1}$ without affecting $F(D)$.

\begin{proposition}[Minimization over $P_{X_1}$] \label{prop:min-over-PX1}
    The rate-cost function $F(D)$ can be expressed as
    \begin{equation} \label{eq:min-wrt-PX1}
        F(D) =
        \lim_{n \to \infty}
        \min_{\substack{ P_{X_1}, \, \left\{P_{U_t \mid X_t, U_{t-1}}\right\}_{t=1}^n \colon \\ \frac{1}{n} \sum_{t=1}^n \EE\left[c\left(X_t, U_t\right)\right] \le D}}
        \dfrac{1}{n} \sum_{t=1}^n I(X_t ; U_t \mid U_{t-1}) .
    \end{equation}
\end{proposition}

\begin{proof}[Proof sketch]
    By~\eqref{eq:simplified}, the direction $\ge$ is immediate, since fixing $X_1 \sim P_1$ is more restrictive than minimizing over $P_{X_1}$.
    For the direction $\le$,
    we show that for any fixed feasible collection of distributions $\Bigl\{ P'_{X_1}, \bigl\{P'_{U_t \mid X_t, U_{t-1}}\bigr\}_{t=1}^n \Bigr\}$,
    there exists a feasible collection of distributions $\bigl\{P_{U_t \mid X_t, U_{t-1}}\bigr\}_{t=1}^n$
    that has the same normalized objective and cost in the limit $n\to\infty$.
    This yields the equality of the two limits, and hence~\eqref{eq:min-wrt-PX1}.
    The full proof is given in Section~\ref{subsec:proof-prop:min-over-PX1}.
\end{proof}

The optimization in~\eqref{eq:min-wrt-PX1} is naturally parametrized by conditional distributions\linebreak
$\left( P_{X_1} , \left\{P_{U_t \mid X_t,U_{t-1}} \right\}_{t=1}^n \right)$.
For the symmetry arguments that follow, it is more convenient to reparametrize the problem in terms of the joint distributions of $(U_{t-1},X_t,U_t)$.

We introduce a \emph{marginalization} operator $\Marg$.
For a finite index set $A$, let $Z_A \coloneq (Z_i)_{i\in A}$
and $z_A \coloneq (z_i)_{i\in A}$.
For $S \subseteq A$, given a joint probability distribution $P_{Z_A}$,
we write $\Marg_{Z_S}\!\left(P_{Z_A}\right)$ for the distribution obtained
by marginalizing $P_{Z_A}$ onto the subcollection $Z_S$, namely,
\begin{equation} \label{eq:def-marg}
    \bigl[\Marg_{Z_S}\!\left(P_{Z_A}\right)\bigr](z_S)
    \coloneq
    \sum\nolimits_{z_{A\setminus S}}
    P_{Z_A}(z_S,z_{A\setminus S}),
\end{equation}
where the sum is over all values of $z_{A\setminus S}$.
Thus, $\Marg_{Z_S}(P_{Z_A})=P_{Z_S}$.
As an illustration, if $A=\{1,2,3,4\}$ and $S=\{2,3\}$, then
\begin{equation}
    \Marg_{Z_S}\!\left(P_{Z_A}\right)
    =
    \Marg_{Z_2,Z_3}\!\left(P_{Z_{[4]}}\right)
    =
    P_{Z_2,Z_3} .
\end{equation}

The next proposition performs a change of variables
from the conditionals to the joints,
and makes explicit the marginal consistency across time and the system-state constraints.

\begin{proposition}[Change of variables] \label{prop:c-o-v}
    For each minimization instance in~\eqref{eq:min-wrt-PX1}, we have
    \begin{IEEEeqnarray}{rCl rCl}
        \min_{\substack{
            P_{X_1}, \,
            \left\{P_{U_t \mid X_t, U_{t-1}}\right\}_{t=1}^n \colon \\
            \frac{1}{n} \sum_{t=1}^n
            \EE\left[c\left(X_t, U_t\right)\right] \le D
        }}
        \!\! \sum_{t=1}^n \! I(X_t ; U_t | U_{t-1})
        & = \! &
        \IEEEeqnarraymulticol{4}{l}{
            \min_{
                P_{U_1, X_1}, \,
                \left\{P_{U_t, X_t, U_{t-1}}\right\}_{t=2}^n
            }
            \sum_{t=1}^n \! I(X_t ; U_t \mid U_{t-1})
        }
        \label{eq:lem-c-o-v}
        \\[-2ex]
        &
        \rlap{\textup{s.t.}}
        & \ &
        \frac{1}{n} \sum_{t=1}^n
        \EE\left[c\left(X_t, U_t\right)\right] \!
        & \le &
        \! D ;
        \label{eq:c-o-v-constr0}
        \\[-0.5ex]
        &&&
        \Marg_{U_t} \!
        \left(P_{U_{t+1}, X_{t+1}, U_t}\right) \!
        & = &
        \! \Marg_{U_t} \!
        \left(P_{U_t, X_t, U_{t-1}}\right) \! ,
        \label{eq:c-o-v-constr1}
        \\[-2ex]
        &&&
        \IEEEeqnarraymulticol{3}{r}{
            t=1,\ldots,n-1 ;
        }
        \notag
        \\[-1.5ex]
        &&&
        P_{X_t \mid U_{t-1}} \!
        & = &
        \! Q_{X \mid U'} , \ t=2,\ldots,n .
        \label{eq:c-o-v-constr2}
    \end{IEEEeqnarray}
\end{proposition}
\noindent
Here the set of constraints~\eqref{eq:c-o-v-constr1} enforces \emph{consistency} of the marginals of the $U_t$'s across time, and the set~\eqref{eq:c-o-v-constr2} enforces that the system states $X_t$ evolve according to the given kernel $Q_{X \mid U'}$.

\begin{proof}[Proof sketch]
    We prove the equality by an explicit change of variables between the two\linebreak
    parametrizations.
    Starting from any feasible collection of variables for the left-hand side,
    we construct a feasible collection of variables for the right-hand side
    with the same objective and average cost,
    which proves the direction $\ge$.
    By an analogous argument for the other direction, we also show $\le$.
    As a result, the feasible sets on the two sides of~\eqref{eq:lem-c-o-v} correspond bijectively, with both the objective value and the average cost preserved, thereby proving the equality.
    The full proof is given in Section~\ref{subsec:proof-prop:c-o-v}.
\end{proof}

Combining Propositions~\ref{prop:min-over-PX1} and~\ref{prop:c-o-v} yields the following characterization of $F(D)$:
\begin{equation} \label{eq:c-o-v}
    \begin{aligned}
        F(D) =
        \lim_{n\to\infty}
        & \min_{ P_{U_1, X_1}, \, \left\{P_{U_t, X_t, U_{t-1}}\right\}_{t=2}^n } \dfrac{1}{n} \sum_{t=1}^n I(X_t ; U_t \mid U_{t-1}) \\
        & \qquad\qquad \textup{s.t.} \quad \eqref{eq:c-o-v-constr0},\eqref{eq:c-o-v-constr1},\eqref{eq:c-o-v-constr2}.
    \end{aligned}
\end{equation}
Each minimization instance in~\eqref{eq:c-o-v} is almost cyclically invariant;
only its $t=1$ boundary is exceptional,
since it uses $P_{U_1,X_1}$ instead of a triple of the form $P_{U_t,X_t,U_{t-1}}$.
We call by \emph{symmetrization} the operation of closing this
loop by extending $P_{U_1,X_1}$ to $P_{U_1,X_1,U_n}$
with an appropriate marginalization constraint.
This operation is only an auxiliary proof step and does not introduce a noncausal dependence; Proposition~\ref{prop:cyclify} shows that its normalized boundary effect vanishes as $n\to\infty$.
This \emph{symmetrization} is depicted in Figure~\ref{fig:dep3}.

\begin{proposition}[Symmetrization] \label{prop:cyclify}
    The rate-cost function $F(D)$ can be expressed as
    \begin{IEEEeqnarray}{l l r C l}
        F(D) =
        \lim_{n \to \infty}
        &
        \IEEEeqnarraymulticol{4}{l}{
            \min_{
                P_{U_1, X_1, U_n}, \,
                \left\{P_{U_t, X_t, U_{t-1}}\right\}_{t=2}^n
            }
            \dfrac{1}{n} \sum_{t=1}^n
            I(X_t ; U_t \mid U_{t-1})
        }
        \label{eq:symmetric}
        \\
        &
        \textup{s.t.}
        &
        &
        \makebox[0pt][c]{\eqref{eq:c-o-v-constr0}}
        & \; ;
        \\[-1ex]
        &
        &
        \Marg_{ U_t }
        \bigl( P_{U_{t+1}, X_{t+1}, U_t} \bigr)
        & = &
        \Marg_{ U_t }
        \bigl( P_{U_t, X_t, U_{t-1}} \bigr) ,
        \ t=1,\ldots,n ;
        \label{eq:sym-constr1}
        \\[-1ex]
        &
        &
        P_{X_t \mid U_{t-1}}
        & = &
        Q_{X \mid U'} ,
        \ t=1,\ldots,n ,
        \label{eq:sym-constr2}
    \end{IEEEeqnarray}
    where, in the minimization of size $n$, we set $U_0 \equiv U_n$, $U_{n+1} \equiv U_1$, and $X_{n+1} \equiv X_1$.
\end{proposition}

\begin{proof}[Proof sketch]
    We again justify the two directions $\le$ and $\ge$ separately by
    fixing a set of feasible variables for one side
    and then constructing a feasible collection of variables for the other side.
    We show that the normalized objective and cost are preserved in the limit $n\to\infty$ via the vanishing boundary differences due to the normalization $\frac1n$.
    The full proof is given in Section~\ref{subsec:proof-prop:cyclify}.
\end{proof}

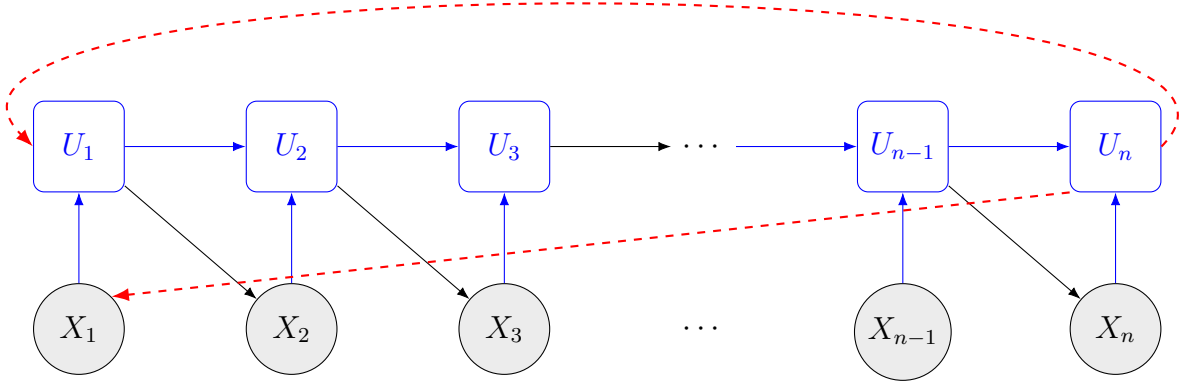
\begin{figure}[!t]
    \centering

    \makebox[\textwidth][c]{%
    \hspace*{-0.1cm}%
    \begin{tikzpicture}[
            >=Latex,
            node distance=12mm and 16mm,
            block/.style={
                draw, rounded corners,
                align=center,
                minimum height=10mm, inner sep=3pt, fill=white
            },
            state/.style={block, minimum width=1.2cm, minimum height=1.2cm},
            ustate/.style={state, draw=blue, text=blue, fill=blue!0},
            line/.style={-Latex},
            uline/.style={-Latex, blue}
        ]

        \node[ustate] (U1) {$U_1$};
        \node[ustate, right=of U1] (U2) {$U_2$};
        \node[ustate, right=of U2] (U3) {$U_3$};
        \node[right=of U3] (Udots) {$\cdots$};
        \node[ustate, right=of Udots] (Unm1) {$U_{n-1}$};
        \node[ustate, right=of Unm1] (Un) {$U_n$};

        \node[state, circle, fill=gray!15, below=of U1] (X1) {$X_1$};
        \node[state, circle, fill=gray!15, below=of U2] (X2) {$X_2$};
        \node[state, circle, fill=gray!15, below=of U3] (X3) {$X_3$};
        \node[right=of X3] (Xdots) {$\cdots$};
        \node[state, circle, fill=gray!15, below=of Unm1] (Xnm1) {$X_{n-1}$};
        \node[state, circle, fill=gray!15, below=of Un] (Xn) {$X_n$};

        \draw[uline] (X1) -- (U1);
        \draw[uline] (X2) -- (U2);
        \draw[uline] (X3) -- (U3);
        \draw[uline] (Xnm1) -- (Unm1);
        \draw[uline] (Xn) -- (Un);

        \draw[uline] (U1) -- (U2);
        \draw[uline] (U2) -- (U3);
        \draw[line] (U3) -- (Udots);
        \draw[uline] (Udots) -- (Unm1);
        \draw[uline] (Unm1) -- (Un);

        \draw[line] (U1) -- (X2);
        \draw[line] (U2) -- (X3);
        \draw[line] (Unm1) -- (Xn);

        \draw[line, red, thick, dashed]
            (Un.east) .. controls +(25mm,25mm) and +(-25mm,25mm) .. (U1.west);
        \draw[line, red, thick, dashed] (Un.south west) -- (X1.north east);

    \end{tikzpicture}%
    }

    \caption{
        Symmetrization of the finite-horizon dependency graph.
        The auxiliary dashed red edges $U_n \to U_1$
        and $U_n \to X_1$
        close the chain and make the formulation invariant under cyclic shifts of the time index.
    }
    \label{fig:dep3}
\end{figure}

After symmetrization, every time index has the same role.
Each minimization instance in \eqref{eq:symmetric} is convex~\cite{Charalambous2016DirectedInformation, atay2026ratecost},
and \emph{exactly} invariant under cyclic shifts of the time index.
We
therefore apply symmetry reduction to pass from the $n$-letter formulation to a single-letter one,
yielding Theorem~\ref{thm:single-letter}
and giving the characterization in~\eqref{eq:single-letter} for $F(D)$.
The full proof of Theorem~\ref{thm:single-letter} is given in Section~\ref{subsec:proof-thm:single-letter}.

Figure~\ref{fig:dep4} illustrates this symmetry reduction at the level of the dependency graph.
In particular, it may be viewed as the single-letter dependency block of the symmetrized graph in Figure~\ref{fig:dep3}, with $(U',X,U)$ representing the common distribution of $(U_{t-1},X_t,U_t)$ under cyclic invariance.

\begin{figure}[t]
    \centering

    \begin{tikzpicture}[
            >=Latex,
            node distance=18mm and 22mm,
            block/.style={
                draw, rounded corners,
                align=center,
                minimum height=10mm, inner sep=4pt, fill=white
            },
            state/.style={block, minimum width=1.2cm, minimum height=1.2cm},
            ustate/.style={state, draw=blue, text=blue, fill=blue!0},
            line/.style={-Latex},
            uline/.style={-Latex, blue}
        ]

        \node[ustate] (Up) {$U'$};
        \node[ustate, right=of Up] (U) {$U$};

        \node[state, circle, fill=gray!15, below=of U] (X) {$X$};

        \draw[uline] (Up) -- (U);
        \draw[line] (Up) -- (X);
        \draw[uline] (X) -- (U);

    \end{tikzpicture}

    \caption{
        Graph describing
        the
        single-letter dependency structure obtained after symmetry reduction.
        It represents the single-letter dependency block of the symmetrized
        dependency graph in Figure~\ref{fig:dep3}, with $(U',X,U)$ representing the
        common distribution of $(U_{t-1},X_t,U_t)$ across time.
    }
    \label{fig:dep4}
\end{figure}
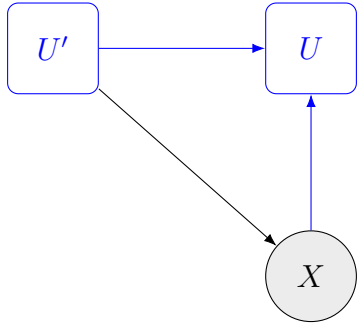

As a consequence, the higher-order-memory systems described in Remark~\ref{rem:higher-order-memory} admit the same single-letter characterization after replacing $U$ and $U'$ in Theorem~\ref{thm:single-letter} by the augmented actions $\widetilde U$ and $\widetilde U'$, respectively.

The derivation of the single-letter expression for the rate-cost function in Theorem~\ref{thm:single-letter} closely parallels the methodology developed in~\cite[Sec.~2]{Atay2026PosetMarkov}. The main technical differences stem from the slight distinction between the mathematical formulations of the rate-cost and channel-capacity problems: in~\cite[Sec.~2]{Atay2026PosetMarkov}, establishing independence from the initial distribution was more delicate and required fixed-point theory for set-valued maps, whereas here that step is simpler; on the other hand, in the present setting Proposition~\ref{prop:cyclify} yields
an asymptotic formulation of $F(D)$
where each optimization problem instance is
exactly symmetric,
while in~\cite[Sec.~2]{Atay2026PosetMarkov} the converse and achievability were proved separately using an approximately symmetric form.

\section{Proofs of Theorem~\ref{thm:single-letter} and Intermediate Propositions} \label{sec:proofs}

\subsection{Proof of Theorem~\ref{thm:single-letter}} \label{subsec:proof-thm:single-letter}

In the characterization of $F(D)$ given in~\eqref{eq:symmetric},
fix a time horizon $n$.
Since the minimization
instance on the right-hand side (RHS) of size $n$
is convex and invariant under cyclic shifts of $\left( X_{[n]},U_{[n]} \right)$, symmetry reduction~\cite[Ex.~4.4]{Boyd2004convex} implies that an optimal solution may be chosen to be cyclically invariant, as follows.
Averaging over all time-index shifts preserves feasibility and cannot increase the objective by convexity,
so an optimizer may be chosen with $P_{U_t,X_t,U_{t-1}}=P_{U,X,U'}$ for all $t$.
The objective then reduces to
\begin{equation}
    \dfrac{1}{n}\sum_{t=1}^n I(X_t ; U_t \mid U_{t-1}) = I \! \left( X;U\mid U' \right) ,
\end{equation}
and the constraints then reduce to $\EE[c(X,U)] \le D$, $P_U = P_{U'}$, and $P_{X\mid U'}=Q_{X\mid U'}$, which are exactly as in~\eqref{eq:single-letter}.
Hence each finite-horizon minimization equals the right-hand side of~\eqref{eq:single-letter}, and therefore so does the limit as $n \to \infty$.
This yields
Theorem~\ref{thm:single-letter}.

\subsection{Proof of Proposition~\ref{prop:simplification}} \label{subsec:proof-prop:simplification}

We justify the directions LHS $\ge$ RHS and LHS $\le$ RHS separately.
    
\underline{For LHS $\ge$ RHS,}
fix a set of feasible variables $\left\{P'_{U_t \mid X_{[t]}, U_{[t-1]}}\right\}_{t=1}^n$ for the LHS.
We construct a set of feasible variables $\left\{P_{U_t \mid X_t, U_{t-1}}\right\}_{t=1}^n$ for the RHS as follows.
Consider the joint distribution $P'_{X_{[n]},U_{[n]}}$
given by
\begin{align}
    P'_{X_{[n]},U_{[n]}} &= \prod_{t=1}^n P'_{U_t \mid X_{[t]}, U_{[t-1]}} \, P'_{X_t \mid X_{[t-1]}, U_{[t-1]}} \\
    &= \prod_{t=1}^n P'_{U_t \mid X_{[t]}, U_{[t-1]}} \, Q_{X_t \mid U_{t-1}} .
\end{align}
Take the distributions $P'_{U_t \mid X_t, U_{t-1}}$ given by
\begin{equation}
    P'_{U_t \mid X_t, U_{t-1}} = \dfrac{P'_{U_t, X_t, U_{t-1}}}{P'_{X_t, U_{t-1}}}
\end{equation}
where the numerator and the denominator are obtained by marginalizing the full joint distribution $P'_{X_{[n]},U_{[n]}}$.
Now for the RHS, take the variables to be
\begin{equation} \label{eq:proofs-1.1}
    P_{U_t \mid X_t, U_{t-1}} \coloneq P'_{U_t \mid X_t, U_{t-1}}, \quad t=1,\ldots,n.
\end{equation}
Let us show via induction that the induced triple-joints are equal for the LHS and the RHS, i.e.,
\begin{equation} \label{eq:proofs-1.2}
    P_{U_t, X_t, U_{t-1}} = P'_{U_t, X_t, U_{t-1}}, \quad t=1,\ldots,n .
\end{equation}
The base case $t=1$ holds because $P_{X_1} = P_1$ is fixed and we have set $P_{U_1 \mid X_1} = P'_{U_1 \mid X_1}$ which together imply $P_{U_1, X_1} = P'_{U_1, X_1}$. For the inductive step, assume we have $P_{U_t, X_t, U_{t-1}} = P'_{U_t, X_t, U_{t-1}}$ for some $t\ge1$. Marginalizing this joint gives $P_{U_t} = P'_{U_t}$, which together with~\eqref{eq:proofs-1.1} yields
\begin{align}
    P_{U_{t+1}, X_{t+1}, U_t} &= P_{U_{t+1} \mid X_{t+1}, U_t} \, Q_{X_{t+1} \mid U_t} \, P_{U_t} \\
    &= P'_{U_{t+1} \mid X_{t+1}, U_t} \, Q_{X_{t+1} \mid U_t} \, P'_{U_t} \\
    &= P'_{U_{t+1}, X_{t+1}, U_t} .
\end{align}
This concludes the induction, and hence shows~\eqref{eq:proofs-1.2}.
Furthermore, the constructed\linebreak
$\left\{P_{U_t \mid X_t, U_{t-1}}\right\}_{t=1}^n$ for the RHS is feasible as the equality $P_{U_t, X_t} = P'_{U_t, X_t}$ implies that the cost remains the same.
Next, we observe that the induced objective for the LHS is greater than or equal to the induced objective for the RHS.
Denote by $I'(\cdot)$ and $I(\cdot)$ the conditional mutual information terms computed with $P'_{X_{[n]}, U_{[n]}}$ and $P_{X_{[n]}, U_{[n]}}$, respectively.
For each summand on the LHS, using the chain rule we have
\begin{align}
    I' \! \left( X_{[t]} ; U_t \, \big| \, U_{[t-1]} \right) 
    &= I' \! \left( X_t ; U_t \, \big| \, X_{[t-1]}, U_{[t-1]} \right)  + \underbrace{I' \! \left( X_{[t-1]} ; U_t \, \big| \, U_{[t-1]} \right) }_{\ge 0} \\
    &\ge I' \! \left( X_t ; U_t \, \big| \, X_{[t-1]}, U_{[t-1]} \right)  \\
    &= I' \! \left( X_t ; U_t \, \big| \, U_{t-1} \right) 
       + \underbrace{I' \! \left( X_t ; X_{[t-1]},U_{[t-2]} \, \big| \, U_t, U_{t-1} \right)}_{\ge 0} \notag \\
    &\qquad
       - \underbrace{I' \! \left( X_t ; X_{[t-1]},U_{[t-2]} \, \big| \, U_{t-1} \right)}_{=0} \label{eq:chain} \\
    &\ge I'(X_t ; U_t \mid U_{t-1}) \\
    &= I(X_t ; U_t \mid U_{t-1}) \label{eq:proofs-1.3}
\end{align}
where~\eqref{eq:chain} follows from $I(A;B\mid C,D)=I(A;B\mid C)+I(A;D\mid B,C)-I(A;D\mid C)$,
and~\eqref{eq:proofs-1.3} follows from~\eqref{eq:proofs-1.2}.
Summing over $t=1,\ldots,n$ gives
\begin{equation}
    \sum_{t=1}^n I' \! \left( X_{[t]} ; U_t \, \big| \, U_{[t-1]} \right) \ge \sum_{t=1}^n I(X_t ; U_t \mid U_{t-1}) .
\end{equation}
This means that for any set of feasible variables for the LHS of~\eqref{eq:simplification}, there exists a set of feasible variables for the RHS with objective less than or equal to that of the LHS, which implies LHS $\ge$ RHS in~\eqref{eq:simplification}.

\underline{For LHS $\le$ RHS,}
fix a set of feasible variables $\left\{P'_{U_t \mid X_t, U_{t-1}}\right\}_{t=1}^n$ for the RHS.
We construct a set of feasible variables $\left\{P_{U_t \mid X_{[t]}, U_{[t-1]}}\right\}_{t=1}^n$ for the LHS as
\begin{align}
    P_{U_t \mid X_{[t]}, U_{[t-1]}} &= P_{U_t \mid X_t, U_{t-1}} \label{eq:lhs-construct} \\
    &= P'_{U_t \mid X_t, U_{t-1}} .
\end{align}
The inductive proof for
the equality~\eqref{eq:proofs-1.2}
of the triple-joint distributions we gave in the\linebreak
LHS $\ge$ RHS part above also applies here verbatim. Hence, we have $P_{U_t, X_t, U_{t-1}} = P'_{U_t, X_t, U_{t-1}}$ for $t=1,\ldots,n$. This, in particular, implies $P_{U_t, X_t} = P'_{U_t, X_t}$ and therefore shows that the constructed set of variables $\left\{P_{U_t \mid X_{[t]}, U_{[t-1]}}\right\}_{t=1}^n$ for the LHS is feasible as the cost remains the same as in the RHS.
Next, to compare the objective values on the LHS and RHS, we denote the terms on the LHS by $I(\cdot)$ as usual and the ones on the RHS by $I'(\cdot)$.

Multiplying the two equalities~\eqref{eq:static-2}
and~\eqref{eq:lhs-construct} side-by-side implies
\begin{equation}
    P_{X_t,U_t\mid X_{[t-1]},U_{[t-1]}} = P_{X_t, U_t \mid U_{t-1}} .
\end{equation}
Therefore, we have the conditional independence
\begin{equation} \label{eq:pair-ci}
    (X_t,U_t)\perp \left( X_{[t-1]},U_{[t-2]} \right) \,\big|\, U_{t-1}.
\end{equation}
Also, since, by~\eqref{eq:lhs-construct}, the control kernel depends on the past only through $(X_t,U_{t-1})$, we have
\begin{equation} \label{eq:ut-ci-given-current}
    U_t \perp \left( X_{[t-1]},U_{[t-2]} \right) \,\big|\, X_t,U_{t-1}.
\end{equation}
The independencies \eqref{eq:pair-ci} and \eqref{eq:ut-ci-given-current} have the following consequences:
\begin{align}
    U_t &\perp X_{[t-1]} \,\big|\, X_t,U_{[t-1]}, \label{eq:ci1} \\
    U_t &\perp U_{[t-2]} \,\big|\, U_{t-1}, \label{eq:ci2} \\
    U_t &\perp U_{[t-2]} \,\big|\, X_t,U_{t-1}. \label{eq:ci3}
\end{align}
Using these conditional independencies, we obtain
\begin{align}
    I \! \left( X_{[t]} ; U_t \, \big| \, U_{[t-1]} \right)
    &= I \! \left( X_t;U_t \, \big| \, U_{[t-1]} \right) + \underbrace{I \! \left( X_{[t-1]};U_t \, \big| \, X_t,U_{[t-1]} \right)}_{=0} \label{eq:proofs-1.4} \\
    &= I \! \left( X_t;U_t \, \big| \, U_{[t-1]} \right) \\
    &= I \! \left( X_t;U_t \, \big| \, U_{[t-1]} \right) + \underbrace{I \! \left( U_{[t-2]};U_t \, \big| \, U_{t-1} \right)}_{=0} \\
    &= I \! \left( X_t,U_{[t-2]};U_t \, \big| \, U_{t-1} \right) \label{eq:proofs-1.5} \\
    &= I(X_t;U_t\mid U_{t-1}) + \underbrace{I \! \left( U_{[t-2]};U_t \, \big| \, X_t,U_{t-1} \right)}_{=0} \label{eq:proofs-1.6} \\
    &= I(X_t;U_t\mid U_{t-1}) \\
    &= I'(X_t;U_t\mid U_{t-1}), \label{eq:lhs-construct-objective}
\end{align}
where~\eqref{eq:proofs-1.4},~\eqref{eq:proofs-1.5} and~\eqref{eq:proofs-1.6} follow from the chain rule;
and~\eqref{eq:lhs-construct-objective} is due to the constructed law $P$ and the original law $P'$ having the same triple-joint law as in~\eqref{eq:proofs-1.2}.

By~\eqref{eq:lhs-construct-objective}, we conclude that for any set of feasible variables for the RHS of~\eqref{eq:simplification}, there exists a set of feasible variables for the LHS with objective equal to that of the RHS, which implies LHS $\le$ RHS in~\eqref{eq:simplification}.

As a result, we have LHS $\le$ RHS as well as LHS $\ge$ RHS in~\eqref{eq:simplification}, which implies equality, and therefore yields Proposition~\ref{prop:simplification}.

\subsection{Proof of Proposition~\ref{prop:min-over-PX1}} \label{subsec:proof-prop:min-over-PX1}

Throughout the proof, the quantities $F_n(D)$ and $\widetilde F_n(D)$ denote
\begin{align}
    F_n(D)
    \coloneq
    & \min_{\substack{ \left\{P_{U_t \mid X_t, U_{t-1}}\right\}_{t=1}^n \\ \frac{1}{n} \sum_{t=1}^n \EE\left[c\left(X_t, U_t\right)\right] \le D}} \dfrac{1}{n} \sum_{t=1}^n I(X_t ; U_t \mid U_{t-1}) \label{eq:proofs-2.1} \\
    & \qquad\quad \textup{s.t.} \quad X_1 \sim P_1 \notag
\end{align}
and
\begin{align}
    \widetilde F_n(D)
    \coloneq
    \min_{\substack{ P_{X_1}, \, \left\{P_{U_t \mid X_t, U_{t-1}}\right\}_{t=1}^n \\ \frac{1}{n} \sum_{t=1}^n \EE\left[c\left(X_t, U_t\right)\right] \le D}} \dfrac{1}{n} \sum_{t=1}^n I(X_t ; U_t \mid U_{t-1}) .
\end{align}
By~\eqref{eq:simplified}, it is necessary and sufficient to prove the equality
\begin{equation} \label{eq:nec-and-suff-proof}
    \limsup_{n \to \infty} F_n(D)
    =
    \limsup_{n \to \infty} \widetilde F_n(D) .
\end{equation}

Since fixing $X_1\sim P_1$ is more restrictive than minimizing over
$P_{X_1}$, we have
\begin{equation}
    \widetilde F_n(D)\le F_n(D)
    \quad \forall n.
\end{equation}
Therefore, using~\eqref{eq:simplified},
\begin{equation} \label{eq:tilde-limsup-upper}
    \limsup_{n\to\infty}\widetilde F_n(D)
    \le
    \lim_{n\to\infty}F_n(D)
    =
    F(D).
\end{equation}
This proves the direction LHS $\ge$ RHS.

For the reverse direction LHS $\le$ RHS, fix $n \ge 1$ and take any feasible variables\linebreak
$\Bigl\{ P'_{X_1}, \, \bigl\{P'_{U_t \mid X_t, U_{t-1}}\bigr\}_{t=1}^n \Bigr\}$
for $\widetilde F_n(D)$.
Let
\begin{equation}
    P'_{U_1}(\cdot)
    =
    \sum_{x_1 \in \mc X} P'_{U_1 \mid X_1}(\cdot \mid x_1)\, P'_{X_1}(x_1)
\end{equation}
be the induced marginal of $U_1$.
Now define variables for $F_n(\cdot)$ by
\begin{align}
    P_{U_1 \mid X_1}
    &\coloneq P'_{U_1}, \\
    P_{U_t \mid X_t, U_{t-1}}
    &\coloneq P'_{U_t \mid X_t, U_{t-1}},
    \qquad t=2,\ldots,n .
\end{align}
This construction yields $P_{U_2}=P'_{U_2}$, and via induction we see that
\begin{align}
    P_{U_t}=P'_{U_t}
    &\implies P_{X_{t+1}, U_t}=P'_{X_{t+1}, U_t} \label{eq:proofs-2.2} \\
    &\implies P_{U_{t+1}, X_{t+1}, U_t}=P'_{U_{t+1}, X_{t+1}, U_t} \label{eq:proofs-2.3} \\
    &\implies P_{U_{t+1}}=P'_{U_{t+1}} , \label{eq:proofs-2.4}
\end{align}
where~\eqref{eq:proofs-2.2} is obtained by multiplying both sides by $Q_{X_{t+1}\mid U_t}$,
\eqref{eq:proofs-2.3} is because\linebreak
${P_{U_{t+1}\mid X_{t+1},U_t}=P'_{U_{t+1}\mid X_{t+1},U_t}}$,
and~\eqref{eq:proofs-2.4} is by marginalizing.
Hence $P_{U_t, X_t, U_{t-1}} = P'_{U_t, X_t, U_{t-1}}$ for $t=2,\ldots,n$.
Denote by $I'(\cdot)$ and $I(\cdot)$ the conditional mutual information terms computed with distributions $P'$ and $P$, respectively,
and analogously for the expectations $\EE'[\cdot]$ and $\EE[\cdot]$.
It follows that, for every $t\ge2$,
\begin{equation}
    I(X_t;U_t\mid U_{t-1}) = I'(X_t;U_t\mid U_{t-1}),
    \qquad
    \EE[c(X_t,U_t)] = \EE'[c(X_t,U_t)] .
\end{equation}
Also, by construction, $U_1$ is independent of $X_1$ under the law $P$, so
\begin{equation}
    I(X_1;U_1)=0 \le I'(X_1;U_1).
\end{equation}
Therefore,
\begin{equation} \label{eq:obj-compare}
    \frac{1}{n} \sum_{t=1}^n I(X_t;U_t\mid U_{t-1})
    \le
    \frac{1}{n} \sum_{t=1}^n I'(X_t;U_t\mid U_{t-1}) .
\end{equation}

What remains is the cost $c(X_1,U_1)$ at time step $t=1$.
Let $\delta_n$ denote
\begin{equation}
    \delta_n
    \coloneq
    \frac{1}{n} \Bigl| \EE[c(X_1,U_1)]-\EE'[c(X_1,U_1)] \Bigr| .
\end{equation}
Since the costs agree for $t=2,\ldots,n$, the average cost under the law $P$ satisfies
\begin{align}
    \frac{1}{n} \sum_{t=1}^n \EE[c(X_t,U_t)]
    &=
    \frac{1}{n} \sum_{t=1}^n \EE'[c(X_t,U_t)]
    +
    \frac{1}{n}\Bigl(\EE[c(X_1,U_1)]-\EE'[c(X_1,U_1)]\Bigr) \\
    &\le D+\delta_n .
\end{align}
Recall that the alphabets $\mc X, \mc U$ are finite.
Hence, $\Bigl| \EE[c(X_1,U_1)]-\EE'[c(X_1,U_1)] \Bigr|$ can be upper bounded by a constant.
Therefore, since only the single first-stage term is altered, its contribution to the normalized average cost vanishes, i.e., $\delta_n \to 0$ as $n\to\infty$.

Fix any $\varepsilon>0$.
Then we have $\delta_n \le \varepsilon$ for all sufficiently large $n$,
and hence the variables constructed above satisfy the cost constraint in the definition of $F_n(D+\varepsilon)$;
that is, they form a feasible point of the optimization problem defining $F_n(D+\varepsilon)$.
Combining this with~\eqref{eq:obj-compare} yields, for all sufficiently large $n$,
\begin{equation}
    F_n(D+\varepsilon)
    \le
    \frac{1}{n} \sum_{t=1}^n I'(X_t;U_t\mid U_{t-1}) .
\end{equation}
Taking the infimum over all feasible laws $P'$ implies that for all sufficiently large $n$, we obtain
\begin{equation}
    F_n(D+\varepsilon) \le \widetilde F_n(D) .
\end{equation}
Taking the liminf as $n\to\infty$
and using the definition of $F(D+\varepsilon)$ gives
\begin{equation} \label{eq:tilde-liminf-lower-eps}
    F(D+\varepsilon)
    =
    \lim_{n\to\infty}F_n(D+\varepsilon)
    \le
    \liminf_{n\to\infty}\widetilde F_n(D)
    \quad \forall\varepsilon>0.
\end{equation}
Since $F$ is right-continuous at $D$, letting
$\varepsilon\downarrow0$ in~\eqref{eq:tilde-liminf-lower-eps} yields
\begin{equation} \label{eq:tilde-liminf-lower}
    F(D)
    =
    \lim_{\varepsilon\downarrow0}F(D+\varepsilon)
    \le
    \liminf_{n\to\infty}\widetilde F_n(D).
\end{equation}
Combining~\eqref{eq:tilde-limsup-upper}
and~\eqref{eq:tilde-liminf-lower}, we obtain
\begin{equation}
    F(D)
    \le
    \liminf_{n\to\infty}\widetilde F_n(D)
    \le
    \limsup_{n\to\infty}\widetilde F_n(D)
    \le
    F(D).
\end{equation}
Therefore, the limit exists and satisfies
\begin{equation}
    \lim_{n\to\infty}\widetilde F_n(D)=F(D),
\end{equation}
which proves~\eqref{eq:nec-and-suff-proof} and yields
Proposition~\ref{prop:min-over-PX1}.

\subsection{Proof of Proposition~\ref{prop:c-o-v}} \label{subsec:proof-prop:c-o-v}

We fix a set of feasible variables for the LHS and construct one for the RHS with the same objective value, which will show LHS $\ge$ RHS, and then vice versa.

\underline{For LHS $\ge$ RHS,} fix a set of feasible variables $P'_{X_1}, \, \left\{P'_{U_t \mid X_t, U_{t-1}}\right\}_{t=1}^n$ for the LHS.
These variables induce the joint distribution $P'_{U_1,X_1}$ through $P'_{U_1 \mid X_1}$ and $P'_{X_1}$. The marginal $P'_{U_1}$ is then obtained from $P'_{U_1,X_1}$. Together with $Q_{X_2 \mid U_1}$ and $P'_{U_2 \mid X_2, U_1}$, one obtains $P'_{U_2, X_2, U_1} = P'_{U_2 \mid X_2, U_1} \cdot Q_{X_2 \mid U_1} \cdot P'_{U_1}$, which can then be marginalized to obtain $P'_{U_2}$.
Following this inductively, once one has the marginal $P'_{U_t}$, together with $Q_{X_{t+1} \mid U_t}$ and $P'_{U_{t+1} \mid X_{t+1}, U_t}$, one obtains the triple-joint distribution
\begin{equation}
    P'_{U_{t+1}, X_{t+1}, U_t} = P'_{U_{t+1} \mid X_{t+1}, U_t} \cdot Q_{X_{t+1} \mid U_t} \cdot P'_{U_t}
\end{equation}
which by marginalizing gives $P'_{U_{t+1}}$ and carries on to the next step.
The set of variables\linebreak
$P'_{U_1, X_1}, \, \left\{P'_{U_t, X_t, U_{t-1}}\right\}_{t=2}^n$ obtained through this process:
(i) induces the same objective value;
(ii) induces the same cost; and
(iii) satisfies the constraints on the RHS of~\eqref{eq:lem-c-o-v} by construction.
Therefore, for each fixed set of feasible variables for the LHS,
one has a set of feasible variables for the RHS with the same objective.
This shows LHS $\ge$ RHS.

\underline{For LHS $\le$ RHS,} fix a set of feasible variables $P'_{U_1, X_1}, \, \left\{P'_{U_t, X_t, U_{t-1}}\right\}_{t=2}^n$ for the RHS of~\eqref{eq:lem-c-o-v} that satisfy the constraints therein.
From each of these joint distributions, one obtains the distributions $P'_{X_1}$, $P'_{U_1 \mid X_1}$, and $P'_{U_t \mid X_t, U_{t-1}}$ for $t=2,\ldots,n$.
Following the procedure of constructing the triple-joints illustrated above in the part LHS $\ge$ RHS,
one recovers the original set of distributions due to the constraints in~\eqref{eq:lem-c-o-v}.
Hence, the set of variables $P'_{X_1}, \, \left\{P'_{U_t \mid X_t, U_{t-1}}\right\}_{t=1}^n$ obtained through this process:
(i) induces the same objective value; and
(ii) induces the same cost.
Therefore, for each fixed set of feasible variables for the RHS,
one has a set of feasible variables for the LHS with the same objective.
This shows LHS $\le$ RHS.

LHS $\ge$ RHS and LHS $\le$ RHS together imply equality in~\eqref{eq:lem-c-o-v}, which yields Proposition~\ref{prop:c-o-v}.

\subsection{Proof of Proposition~\ref{prop:cyclify}} \label{subsec:proof-prop:cyclify}

By~\eqref{eq:lem-c-o-v}, the proposition is equivalent to proving the equality
\begin{equation} \label{eq:cyclify}
    \begin{alignedat}{2}
        \lim_{n\to\infty} \!
        & \min_{\substack{ P_{U_1, X_1}, \\ \left\{P_{U_t, X_t, U_{t-1}}\right\}_{t=2}^n }} \! \dfrac{1}{n} \sum_{t=1}^n I(X_t ; U_t \mid U_{t-1})
        =
        \lim_{n \to \infty} \!
        && \min_{\substack{ P_{U_1, X_1, U_n}, \\ \left\{P_{U_t, X_t, U_{t-1}}\right\}_{t=2}^n }} \! \dfrac{1}{n} \sum_{t=1}^n I(X_t ; U_t \mid U_{t-1}) \\[-1ex]
        & \qquad\ \textup{s.t.} \ \eqref{eq:c-o-v-constr0}, \eqref{eq:c-o-v-constr1}, \eqref{eq:c-o-v-constr2}
        && \qquad\ \textup{s.t.} \ \eqref{eq:c-o-v-constr0}, \eqref{eq:sym-constr1}, \eqref{eq:sym-constr2} .
    \end{alignedat}
\end{equation}

The direction LHS $\le$ RHS is rather straightforward, as follows.
Take any sequence of feasible sets of variables $\Bigl\{ P'_{U_1, X_1, U_n}, \, \bigl\{P'_{U_t, X_t, U_{t-1}}\bigr\}_{t=2}^n \Bigr\}_{n\ge1}$ for the RHS.
For each $n$,
one can marginalize to obtain $P'_{U_1, X_1}$ and can take $\left\{P'_{U_t, X_t, U_{t-1}}\right\}_{t=2}^n$ as is. This yields a feasible set of variables for the LHS.
Moreover, the values of $I(X_t ; U_t \mid U_{t-1})$ agree for $2 \le t \le n$, while for $t=1$ the two objectives differ by at most a constant independent of $n$:
\begin{equation}
    \Bigl| I(X_1;U_1) - I(X_1;U_1 \mid U_n) \Bigr|
    \le H(X_1) \le \log |\mc X| .
\end{equation}
Hence the normalized difference between the two objectives is at most $\frac{\log |\mc X|}{n}$, which vanishes as $n\to\infty$. Therefore, we have LHS $\le$ RHS.

In the rest of this proof, we demonstrate the direction LHS $\ge$ RHS.
Since $\mc X$ and $\mc U$ are finite, the cost is bounded; let
\begin{equation}
    c_{\max} \coloneq \max_{(x,u)\in \mc X\times \mc U} c(x,u) < \infty .
\end{equation}
For all $\varepsilon>0$, define the set $\mc H_\varepsilon$ as
\begin{align} \label{eq:def-D-eps}
    \mc H_\varepsilon
    \coloneq
    \Bigl\{ P_{U,X,U'} \colon \, \EE[c(X,U)] \le D+\varepsilon , \, P_{X \mid U'} = Q_{X \mid U'} , \, \bigl| P_{U}(u) - P_{U'}(u) \bigr| \le \varepsilon \ \forall u \!\in\! \mc U \Bigr\} .
\end{align}

Take any sequence of feasible sets of variables $\Bigl\{ P'_{U_1, X_1}, \, \bigl\{ P'_{U_t, X_t, U_{t-1}} \bigr\}_{t=2}^n \Bigr\}_{n\ge1}$ for the LHS.
For each $n$, define the
distribution $\widetilde P^{(n)}_{U,X,U'}$ as the average of
$\left\{P'_{U_t, X_t, U_{t-1}}\right\}_{t=2}^n$, i.e.,
\begin{equation} \label{eq:tilde-P-def}
    \widetilde P^{(n)}_{U,X,U'} \! \left( u,x,u' \right) \coloneq \dfrac{1}{n-1} \sum_{t=2}^n P'_{U_t, X_t, U_{t-1}} \! \left( u,x,u' \right) \quad \forall \! \left( u,x,u' \right) \in \mc U \times \mc X \times \mc U.
\end{equation}
We have $\widetilde P^{(n)}_{U,X,U'} \in \mc H_{\varepsilon_n}$, where
\begin{equation}
    \varepsilon_n \coloneq \max\left\{ \frac{1}{n-1}, \frac{D+c_{\max}}{n-1} \right\}.
\end{equation}
Indeed, first,
\begin{align}
    \EE_{\widetilde P^{(n)}}[c(X,U)]
    &= \dfrac{1}{n-1}\sum_{t=2}^n \EE[c(X_t,U_t)] \\
    &= \dfrac{n}{n-1}\cdot \dfrac{1}{n}\sum_{t=1}^n \EE[c(X_t,U_t)] - \dfrac{1}{n-1}\EE[c(X_1,U_1)] \\
    &\le \dfrac{n}{n-1}D + \dfrac{c_{\max}}{n-1} \\
    &= D + \dfrac{D+c_{\max}}{n-1} \\
    &\le D+\varepsilon_n .
\end{align}
Next,
due to the constraints on the LHS of~\eqref{eq:cyclify}, we have
$P'_{X_t \mid U_{t-1}} = Q_{X \mid U'}$ for every $t=2,\ldots,n$.
The averaging in~\eqref{eq:tilde-P-def} preserves this property and yields
\begin{equation}
    \widetilde P^{(n)}_{X \mid U'} = Q_{X \mid U'} .
\end{equation}
Finally,
\begin{align}
    \widetilde P^{(n)}_{U} (u) &= \sum_{x,u'} \widetilde P^{(n)}_{U,X,U'} \! \left( u,x,u' \right) \\
    &= \dfrac{1}{n-1} \sum_{t=2}^n P'_{U_t} (u) ,
\end{align}
and analogously
\begin{align}
    \widetilde P^{(n)}_{U'} \! \left( u' \right) &= \sum_{u,x} \widetilde P^{(n)}_{U,X,U'} \! \left( u,x,u' \right) \\
    &= \dfrac{1}{n-1} \sum_{t=2}^n P'_{U_{t-1}} \! \left( u' \right) \\
    &= \dfrac{1}{n-1} \sum_{t=1}^{n-1} P'_{U_t} \! \left( u' \right) ,
\end{align}
which together yield
\begin{align}
    \left| \widetilde P^{(n)}_{U} (u) - \widetilde P^{(n)}_{U'} (u) \right|
    &= \dfrac{1}{n-1} \left| P'_{U_n} (u) - P'_{U_1} (u) \right| \\
    &\le \dfrac{1}{n-1} \quad \forall u\in\mc U .
\end{align}
Hence, by~\eqref{eq:def-D-eps}, we have $\widetilde P^{(n)}_{U,X,U'} \in \mc H_{\varepsilon_n}$.

Again, denote by $I'(\cdot)$ and $\widetilde I(\cdot)$
the conditional mutual information terms computed with distributions $P'$ and $\widetilde P^{(n)}$, respectively.
By the convexity of mutual information and by Jensen's inequality,
we obtain that for each $n$,
\begin{equation}
    \dfrac{1}{n-1}\sum_{t=2}^n I'(X_t ; U_t \mid U_{t-1})
    \ge
    \widetilde I \! \left( X;U\mid U' \right) .
\end{equation}
Therefore,
\begin{equation}
    \dfrac{1}{n}\sum_{t=1}^n I'(X_t;U_t\mid U_{t-1})
    \ge
    \dfrac{n-1}{n}\,\widetilde I \! \left( X;U\mid U' \right) .
\end{equation}
Moreover, each
set of original distributions
$\Bigl\{ P'_{U_1, X_1}, \, \bigl\{ P'_{U_t, X_t, U_{t-1}} \bigr\}_{t=2}^n \Bigr\}_{n\ge1}$ we had fixed for the LHS
can be used to construct the average $\widetilde P^{(n)}_{U,X,U'} \in \mc H_{\varepsilon_n}$ as in~\eqref{eq:tilde-P-def}, which implies that
\begin{equation} \label{eq:ineq-D-1-n}
    \begin{aligned}
        & \min_{ P_{U_1, X_1}, \, \left\{P_{U_t, X_t, U_{t-1}}\right\}_{t=2}^n } \dfrac{1}{n} \sum_{t=1}^n I(X_t ; U_t \mid U_{t-1})
        \ \ge \
        \dfrac{n-1}{n}\min_{ P_{U, X, U'} \in \mc H_{\varepsilon_n}} I \! \left( X;U\mid U' \right) . \\[-1.5ex]
        & \qquad\quad\ \ \textup{s.t.} \ \eqref{eq:c-o-v-constr0}, \eqref{eq:c-o-v-constr1}, \eqref{eq:c-o-v-constr2}
    \end{aligned}
\end{equation}

Now $\varepsilon_n \xrightarrow[n\to\infty]{} 0$.
Since $\mc X, \mc U$ are finite, the probability simplex on $\mc U \times \mc X \times \mc U$ is compact and the functional $P_{U,X,U'} \mapsto I \! \left( X;U\mid U' \right)$ is continuous.
Consider the set $\mc H_0$ defined by
\begin{equation}
    \mc H_0 \coloneq \bigcap_{\varepsilon>0} \mc H_\varepsilon
    = \left\{ P_{U,X,U'} \colon \EE[c(X,U)] \le D,\; P_{X\mid U'}=Q_{X\mid U'},\; P_U=P_{U'} \right\} .
\end{equation}
The sequence $\left\{ \mc H_{\varepsilon_n} \right\}_{n\ge1}$ is a decreasing sequence of nonempty compact sets. Hence, the limit $\mc H_0$ exists and is nonempty.
Thus, we have
\begin{equation}
    \lim_{n\to\infty} \min_{P_{U,X,U'} \in \mc H_{\varepsilon_n}} I \! \left( X;U\mid U' \right)
    =
    \min_{P_{U,X,U'} \in \mc H_0} I \! \left( X;U\mid U' \right) .
\end{equation}
Therefore, taking $\lim_{n\to\infty}$ in~\eqref{eq:ineq-D-1-n} yields
\begin{equation}
    \textup{LHS}
    \ge
    \min_{ P_{U, X, U'} \in \mc H_0} I \! \left( X;U\mid U' \right) .
\end{equation}

But notice that any distribution $P_{U, X, U'} \in \mc H_0$, when taken equal to all $P_{U_t,X_t,U_{t-1}}$ for $t=2,\ldots,n$ and also to $P_{U_1,X_1,U_n}$, yields a feasible point for every minimization instance for the RHS of~\eqref{eq:cyclify}.
Indeed, the condition $P_U=P_{U'}$
gives~\eqref{eq:sym-constr1},
the condition $P_{X\mid U'}=Q_{X\mid U'}$
gives~\eqref{eq:sym-constr2}
and the condition $\EE[c(X,U)]\le D$ gives the cost constraint.
Hence
\begin{equation} \label{eq:D0-ge-RHS}
    \begin{aligned}
        \min_{ P_{U, X, U'} \in \mc H_0} I \! \left( X;U\mid U' \right)
        \ \ge \
        \lim_{n \to \infty}
        & \min_{ P_{U_1, X_1, U_n}, \, \left\{P_{U_t, X_t, U_{t-1}}\right\}_{t=2}^n } \dfrac{1}{n} \sum_{t=1}^n I(X_t ; U_t \mid U_{t-1}) \\[-1.5ex]
        & \qquad\qquad\ \textup{s.t.} \ \eqref{eq:c-o-v-constr0}, \eqref{eq:sym-constr1}, \eqref{eq:sym-constr2} .
    \end{aligned}
\end{equation}
Therefore, LHS $\ge$ RHS in~\eqref{eq:cyclify}.

LHS $\le$ RHS and LHS $\ge$ RHS together imply equality in~\eqref{eq:cyclify}, which yields Proposition~\ref{prop:cyclify}.

\section{Numerical Example} \label{sec:numer}

In this section,
we provide a numerical example
illustrating
the single-letter characterization of the STIC system rate-cost function $F(D)$ given in Theorem~\ref{thm:single-letter}.
Throughout the section, alphabets are $\mc X=\mc U=\{0,1\}$ and the control cost function is $c(x,u)=\mathbf 1\{x\ne u\}$.

The single-letter characterization in Theorem~\ref{thm:single-letter} includes the stationarity constraint $P_U=P_{U'}$.
Thus, in the infinite-horizon limit, the optimal one-step policy must preserve the marginal distribution of the control action.
A controller that optimizes each stage locally, without enforcing this marginal consistency, can therefore be suboptimal over the long run.
This motivates the relaxed quantity
\begin{equation}
    \label{eq:myopic}
    \begin{array}{r@{\;}l@{\quad}r@{\;}c@{\;}l}
        F^\textup{myopic}(D) \coloneq
        &
        \multicolumn{4}{l}{
            \displaystyle
            \min_{P_{U,X,U'}}
            I\!\left(X;U\mid U'\right)
        }
        \\[-1ex]
        &
        \quad \textup{s.t.}
        &
        \displaystyle
        \EE\!\left[c(X,U)\right]
        & \le &
        D \, ,
        \\[-1ex]
        &
        &
        P_{X\mid U'}
        & = &
        Q_{X\mid U'} \, .
    \end{array}
\end{equation}
Since~\eqref{eq:myopic} removes the constraint $P_U=P_{U'}$ from the characterization of $F(D)$,
we have
\begin{equation}
    F^\textup{myopic}(D) \le F(D)
\end{equation}
for every kernel $Q_{X \mid U'}$ and every $D\ge0$.

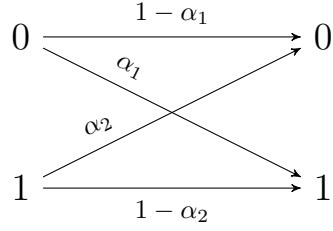
\begin{figure}[t]
    \centering

    \begin{tikzpicture}[>=stealth',
                    node font=\large,
                    every node/.style={align=center},
                    x=4cm, y=2cm]
        \node (L0) at (0,1) {$0$};
        \node (R0) at (1,1) {$0$};
        \node (L1) at (0,0) {$1$};
        \node (R1) at (1,0) {$1$};
        
        \draw[->] (L0) -- node[above,pos=0.5,font=\small] {$1-\alpha_1$} (R0);

        \draw[->] (L0) -- node[pos=0.3,sloped,above,font=\small] {$\alpha_1$} (R1);
        
        \draw[->] (L1) -- node[below,pos=0.5,font=\small] {$1-\alpha_2$} (R1);
        
        \draw[->] (L1) -- node[pos=0.25,sloped,above,font=\small] {$\alpha_2$} (R0);
    \end{tikzpicture}

    \caption{Transition diagram for the kernel $Q_{X \mid U'}$ defined in~\eqref{eq:Z-S}.}
    \label{fig:alpha}
\end{figure}

To illustrate the role of the stationarity constraint, we compare $F(D)$ and
$F^\textup{myopic}(D)$ for a binary STIC system with transition kernel
\begin{equation} \label{eq:Z-S}
    Q_{X \mid U'} (x\mid u')
    \coloneq
    \begin{cases}
        1-\alpha_1, & x=0,\ u'=0,\\[-2ex]
        \alpha_1, & x=1,\ u'=0,\\[-2ex]
        \alpha_2, & x=0,\ u'=1,\\[-2ex]
        1-\alpha_2, & x=1,\ u'=1.
    \end{cases}
\end{equation}
The transition probabilities are shown in Figure~\ref{fig:alpha}.
For $(\alpha_1,\alpha_2)=(0.4,0.3)$, Figure~\ref{fig:alpha-plot} plots
$F(D)$ and $F^\textup{myopic}(D)$ over $D\in[0,0.4]$.
The curves show that the relaxation can be strict:
$F^\textup{myopic}(D)$ lies strictly below $F(D)$ for $D\in[0,0.3)$.
Thus, enforcing the stationary action marginal is essential for characterizing the true infinite-horizon rate-cost tradeoff.

\begin{figure}[t]
    \centering
    \includegraphics[width=1.00\linewidth]{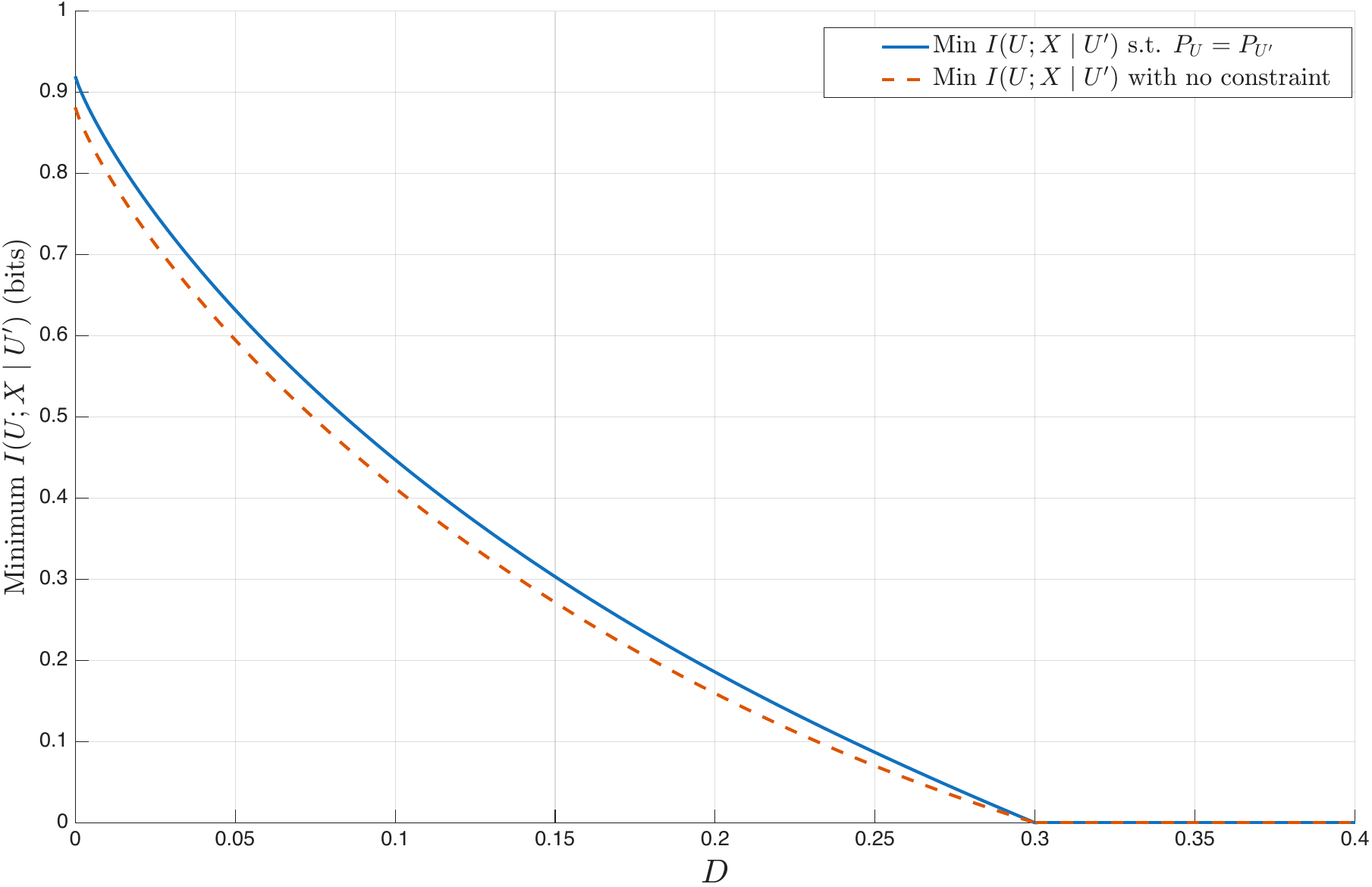}

    \caption{Comparison of the rate-cost function $F(D)$ and its myopic relaxation $F^\textup{myopic}(D)$ for the STIC system defined in~\eqref{eq:Z-S} with parameters $(\alpha_1,\alpha_2)=(0.4,0.3)$.}
    \label{fig:alpha-plot}
\end{figure}

We also observe that $F(D) = F^\textup{myopic}(D) = 0$ for $D\ge0.3$.
Indeed, for $D \ge 0.3$, one may choose the stationary zero-rate policy
$P_{U'}(1)=P_U(1)=1$. Then $U$ is deterministic and independent of $X$
given $U'$, so $I(X;U\mid U') = 0$.
Moreover, since $U'=1$, the state satisfies
$Q_{X\mid U'}(0\mid 1)=\alpha_2=0.3$, and under the Hamming cost this gives
\begin{equation}
    \EE[c(X,U)] = \PP[X\ne U] = \PP[X=0\mid U'=1]=0.3 \le D .
\end{equation}
Thus $F(D)=0$ for $D\ge0.3$. Since
$F^\textup{myopic}(D)\le F(D)$ and both objectives are nonnegative, we also have
$F^\textup{myopic}(D)=0$ for $D\ge0.3$.

\section{Dual STIC and NOST Kernel Families} \label{sec:dual}

Throughout this section, alphabets are $\mc X=\mc U=\{0,1\}$ and the control cost function is $c(x,u)=\mathbf 1\{x\ne u\}$.
For $p\in[0,1]$,
we denote the binary entropy of $p$ by $H(p)$,
which is given by
\begin{equation}
    H(p)
    \coloneq
    - p\log p - (1-p)\log(1-p) ,
\end{equation}
with the convention $0\log 0 = 0$.

We revisit a pair of binary kernel families
that are dual to each other: a family $\mc Q^\textup{STIC}$ of STIC system kernels and a family $\mc Q^\textup{NOST}$ of NOST channel kernels.
These families are closely related to the source--channel matching families studied in~\cite{kourtellaris2014thesis, Kourtellaris2015Nonanticipative}.
The duality is in terms of
the optimizer of one problem inducing a kernel in the other family.
Specifically, for $Q^{(1)}_{X\mid U'}\in\mc Q^\textup{STIC}$, let
$P^{\star(1)}_{U,X,U'}$ denote the optimizer of~\eqref{eq:single-letter}.
Likewise, for $Q^{(2)}_{Y\mid X,Y'}\in\mc Q^\textup{NOST}$, let
$P^{\star(2)}_{X,Y'}$ denote the optimizer of~\eqref{eq:nost-c}.
We show that the two families satisfy
\begin{align}
    Q^{(1)}_{X \mid U'} \in \mc Q^\textup{STIC}
    &\implies
    P^{\star(1)}_{U \mid X,U'} \in \mc Q^\textup{NOST},
    \label{eq:impl-1} \\
    Q^{(2)}_{Y \mid X,Y'} \in \mc Q^\textup{NOST}
    &\implies
    P^{\star(2)}_{X \mid Y'} \in \mc Q^\textup{STIC}.
    \label{eq:impl-2}
\end{align}

\subsection{The kernel families}

First, define $\mc Q^\textup{STIC}$ as the family of binary symmetric STIC
kernels. For $\theta\in(0,1/2)$, let
\begin{equation} \label{eq:def-Q-src}
    Q^{(\theta)}_{X \mid U'}(x | u')
    =
    \begin{cases}
        1-\theta, & x=u',\\
        \theta,   & x\ne u'.
    \end{cases}
\end{equation}
Then
\begin{equation}
    \mc Q^\textup{STIC}
    \coloneq
    \left\{
        Q^{(\theta)}_{X \mid U'} : \theta\in(0,1/2)
    \right\}.
\end{equation}

Next, define $\mc Q^\textup{NOST}$ as follows.
Let $\sigma(t)\coloneq(1+2^{-t})^{-1}$. For $a,b>0$, let
$Q^{(a,b)}_{Y\mid X,Y'}$ denote the binary NOST channel satisfying
\begin{align}
    \PP\!\left[Y=X \mid X=Y'\right] &= \sigma(a+b), \label{eq:def-Q-chn-1} \\
    \PP\!\left[Y=X \mid X\ne Y'\right] &= \sigma(a-b). \label{eq:def-Q-chn-2}
\end{align}
Equivalently,
\begin{equation} \label{eq:def-Q-chn}
    Q^{(a,b)}_{Y \mid X,Y'}(y\mid x,y')
    =
    \frac{
        2^{
            a\,\mathbf 1\{y=x\}
            +
            b\,\mathbf 1\{y=y'\}
        }
    }{
        \sum_{\bar y\in\{0,1\}}
        2^{
            a\,\mathbf 1\{\bar y=x\}
            +
            b\,\mathbf 1\{\bar y=y'\}
        }
    } .
\end{equation}
To see the equivalence in~\eqref{eq:def-Q-chn}, fix $x,y'\in\{0,1\}$.
First suppose that
$x=y'$.
Then, for $y=x=y'$,
\begin{equation}
    a\mathbf 1\{y=x\}+b\mathbf 1\{y=y'\}=a+b,
\end{equation}
whereas for the only other binary output $y\ne x$, we also have
$y\ne y'$, and hence the exponent is $0$.
Therefore
\begin{equation}
    Q^{(a,b)}_{Y\mid X,Y'}(x\mid x,x)
    =
    \frac{2^{a+b}}{2^{a+b}+1}
    =
    \frac{1}{1+2^{-(a+b)}}
    =
    \sigma(a+b).
\end{equation}
Now suppose that $x\ne y'$.
Since the output alphabet is binary,
the two possible outputs are $y=x$ and $y=y'$.
If $y=x$, then $y\ne y'$, so the exponent is $a$.
If $y=y'$, then $y\ne x$, so the exponent is $b$.
Hence
\begin{equation}
    Q^{(a,b)}_{Y\mid X,Y'}(x\mid x,y')
    =
    \frac{2^a}{2^a+2^b}
    =
    \frac{1}{1+2^{b-a}}
    =
    \frac{1}{1+2^{-(a-b)}}
    =
    \sigma(a-b).
\end{equation}
Thus the normalized exponential form in \eqref{eq:def-Q-chn} is equivalent
to the two defining transition probabilities in~\eqref{eq:def-Q-chn-1}--\eqref{eq:def-Q-chn-2}.

We define the family $\mc Q^\textup{NOST}$ of NOST channel kernels as
\begin{equation}
    \mc Q^\textup{NOST}
    \coloneq
    \left\{
        Q^{(a,b)}_{Y\mid X,Y'} : a,b>0
    \right\}.
\end{equation}

\subsection{Optimizing kernels}

We now give the optimizer induced by each family.
Consider the STIC kernel $Q^{(\theta)}_{X\mid U'}$ in~\eqref{eq:def-Q-src},
with $\theta\in(0,1/2)$, and fix a cost level $D<\theta$. Define
\begin{equation} \label{eq:zeta}
    \zeta \coloneq \frac{\theta-D}{1-2D}.
\end{equation}
We prove in the \hyperref[app:binary-stic-opt]{Appendix} below that
the optimizer $P^{\star(\theta)}_{U,X,U'}$ of~\eqref{eq:single-letter}
satisfies
\begin{equation} \label{eq:optimizer}
    P^{\star(\theta)}_{U \mid X,U'}(u\mid x,u')
    =
    \frac{
        2^{
            a\,\mathbf 1\{u=x\}
            +
            b\,\mathbf 1\{u=u'\}
        }
    }{
        \sum_{\bar u\in\{0,1\}}
        2^{
            a\,\mathbf 1\{\bar u=x\}
            +
            b\,\mathbf 1\{\bar u=u'\}
        }
    }
\end{equation}
where
\begin{equation}
    a\coloneq\log\frac{1-D}{D},
    \quad
    b\coloneq\log\frac{1-\zeta}{\zeta}.
\end{equation}
Since $a,b>0$, this conditional kernel belongs to $\mc Q^\textup{NOST}$.
Thus,
\begin{equation}
    Q^{(\theta)}_{X\mid U'}\in\mc Q^\textup{STIC}
    \implies
    P^{\star(\theta)}_{U\mid X,U'}\in\mc Q^\textup{NOST}.
\end{equation}

Conversely, consider the NOST channel $Q^{(a,b)}_{Y\mid X,Y'}$
in~\eqref{eq:def-Q-chn}, with $a,b>0$. The optimizer
$P^{\star(a,b)}_{X,Y'}$ of~\eqref{eq:nost-c} satisfies~\cite[Thm.\ 3.48]{kourtellaris2014thesis}
\begin{equation} \label{eq:P-star-ab}
    P^{\star(a,b)}_{X \mid Y'}(x\mid y')
    =
    \begin{cases}
        1-\lambda^\star, & x=y',\\
        \lambda^\star,   & x\ne y',
    \end{cases}
\end{equation}
where
\begin{equation}
    \lambda^\star
    \coloneq
    \frac{
        e^\star - 1 + \sigma(a+b)
    }{
        \sigma(a+b)+\sigma(a-b)-1
    },
\end{equation}
and
\begin{equation}
    e^\star
    \coloneq
    \frac{1}{
        1+
        2^{
            \frac{
                H(\sigma(a-b)) - H(\sigma(a+b))
            }{
                \sigma(a+b)+\sigma(a-b)-1
            }
        }
    } .
\end{equation}
Hence $P^{\star(a,b)}_{X\mid Y'}$ is a binary symmetric STIC kernel, and so
\begin{equation}
    Q^{(a,b)}_{Y\mid X,Y'}\in\mc Q^\textup{NOST}
    \implies
    P^{\star(a,b)}_{X\mid Y'}\in\mc Q^\textup{STIC}.
\end{equation}

As a result, $\mc Q^\textup{STIC}$ and $\mc Q^\textup{NOST}$ are dual in the
sense of~\eqref{eq:impl-1}--\eqref{eq:impl-2}.

\section{Conclusion} \label{sec:conc}

We studied static time-invariant control (STIC) systems
defined in~\eqref{eq:static-2},
i.e., control systems in which the next system state depends on the history only through the most recent control action in a time-invariant (stationary) manner.
The communication from the system to the controller is rate-limited;
providing more information about the state generally reduces control cost, whereas lower-rate communication typically increases it.
This tradeoff is governed by the infinite-horizon rate-cost function $F(D)$~\cite{atay2026ratecost} given in~\eqref{eq:F-D}.
We showed that
$F(D)$---defined as a limit of growing-size optimizations---admits a constant-size reduction
by demonstrating that
the multi-letter directed-information minimization reduces to a single-letter convex program.
We derived the reduction via symmetry reduction~\cite[Ex.~4.4]{Boyd2004convex} with respect to cyclic shifts of the time index by a proof methodology analogous to the one developed in~\cite[Sec.~2]{Atay2026PosetMarkov}.
The resulting characterization in Theorem~\ref{thm:single-letter} expresses $F(D)$ as the minimum single-letter conditional mutual information $I \! \left( X;U\mid U' \right)$ under linear constraints.
This provides the first single-letter characterization of the rate-cost function beyond LQG control~\cite{KostinaHassibi2018, KostinaHassibi2019RateCostControl, TanakaEsfahaniMitter2018_LQGMinDI},
and identifies it as the exact control-side counterpart of the NOST channel feedback capacity formula~\cite{NOST, Atay2026PosetMarkov}.

Future directions include seeking single-letter reductions of the rate-cost
functions of more general families of control systems beyond STIC systems.
A particularly natural next step is to study time-invariant controlled Markov
systems satisfying
$P_{X_{t+1}\mid X_{[t]},U_{[t]}}
=
P_{X_{t+1}\mid X_t,U_t}$,
thereby relaxing the STIC assumption by allowing the next state to depend on
the current state in addition to the current action.
In view of the STIC--NOST duality developed in this paper, it would be interesting to
formulate and analyze control-side counterparts of finite-state
channels with memory, such as the Trapdoor and Ising channels, whose feedback
capacities are known in closed form~\cite{trapdoor,Ising}.
When an exact characterization is unavailable, it would
be worthwhile to investigate whether $Q$-graph constructions, which yield
computable single-letter bounds for unifilar finite-state
channels~\cite{Sabag-upper-bound}, admit a control-side analogue that yields
computable bounds on $F(D)$.

Other directions include studying variants of the present setting in which
the decoder has feedforward (or more general side information).
One limitation of our current methodology based on symmetry reduction is that
it requires the alphabets to be finite; it would also be of interest to look
for single-letterizations for more general alphabets.
In addition, since in many rate-limited control settings the system is
distributed over a network, the natural notion of causality is induced by the
network topology rather than by a single time index.
In such networked systems with spatial structure, the state is distributed
over a graph, and the effect of a control action propagates through the
network topology.
To obtain computable information-theoretic performance limits for networked
control problems with graphical structure, it would also be interesting to
extend the discrete-time framework we studied to control systems over general
directed acyclic graphs, analogous to the poset-Markov channel models
developed in~\cite[Sec.~3]{Atay2026PosetMarkov}.

\appendix
\label{app:binary-stic-opt}

We prove the optimizer stated in~\eqref{eq:optimizer} for the binary
symmetric STIC kernel in~\eqref{eq:def-Q-src}.
Let $\theta\in(0,1/2)$ and let $D\in(0,\theta)$.
The case $D=0$ follows by taking the limit $a\to\infty$ below.

\begin{proposition}
    For the kernel $Q^{(\theta)}_{X\mid U'}$ in~\eqref{eq:def-Q-src}
    and $\zeta$ given in~\eqref{eq:zeta},
    the single-letter problem in~\eqref{eq:single-letter} has value
    \begin{equation}
        H(\theta)-H(D),
    \end{equation}
    and an optimizer is induced by~\eqref{eq:optimizer}.
\end{proposition}

\begin{proof}
    Let $P_{U,X,U'}$ be any feasible distribution in~\eqref{eq:single-letter},
    and set
    \begin{equation}
        \Delta\coloneq \mathbf 1\{X\ne U\}.
    \end{equation}
    Since $P_{X\mid U'}=Q^{(\theta)}_{X\mid U'}$, we have
    $H(X\mid U')=H(\theta)$. Therefore,
    \begin{align}
        I(X;U\mid U')
        &= H(X\mid U')-H(X\mid U,U') \notag \\
        &= H(\theta)-H(\Delta\mid U,U') \notag \\
        &\ge H(\theta)-H(\Delta). \label{eq:app-lb-1}
    \end{align}
    If the expected cost is at most $D$, then
    $\PP[\Delta=1]\le D$. Since $D<1/2$ and the binary entropy function is
    increasing on $[0,1/2]$,
    \begin{equation}
        H(\Delta)\le H(D).
    \end{equation}
    Therefore, every feasible distribution satisfies
    \begin{equation} \label{eq:app-lb-2}
        I(X;U\mid U') \ge H(\theta)-H(D).
    \end{equation}

    It remains to construct a feasible distribution attaining this lower
    bound. Since $D<\theta<1/2$, we have that $\zeta$ defined in~\eqref{eq:zeta} satisfies $\zeta\in(0,1/2)$. Let
    \begin{equation}
        P^\star_{U'}(u')=\frac12,
    \end{equation}
    and define
    \begin{equation}
        P^\star_{U\mid U'}(u\mid u')
        =
        \begin{cases}
            1-\zeta, & u=u',\\
            \zeta,   & u\ne u',
        \end{cases}
        \qquad
        P^\star_{X\mid U}(x\mid u)
        =
        \begin{cases}
            1-D, & x=u,\\
            D,   & x\ne u.
        \end{cases}
    \end{equation}
    Let
    \begin{equation}
        P^\star_{U',U,X}
        =
        P^\star_{U'} P^\star_{U\mid U'} P^\star_{X\mid U}.
    \end{equation}
    Then $P^\star_U=P^\star_{U'}$ because $P^\star_{U'}$ is uniform and
    $P^\star_{U\mid U'}$ is a binary symmetric kernel. Moreover,
    \begin{align}
        \PP^\star[X\ne U']
        &= \PP^\star[U\ne U'] \cdot \PP^\star[X=U]
           +
           \PP^\star[U=U'] \cdot \PP^\star[X\ne U] \notag \\
        &= \zeta(1-D)+(1-\zeta)D \notag \\
        &= D+\zeta(1-2D) \notag \\
        &= \theta.
    \end{align}
    This implies
    \begin{equation}
        P^\star_{X\mid U'}=Q^{(\theta)}_{X\mid U'}.
    \end{equation}
    Moreover,
    \begin{equation}
        \EE^\star[c(X,U)]
        =
        \PP^\star[X\ne U]
        =
        D.
    \end{equation}
    Thus, $P^\star_{U,X,U'}$ is feasible.

    Since $X$ is conditionally independent of $U'$ given $U$ under
    $P^\star$, and $P^\star_{X\mid U}$ is a binary symmetric kernel with
    crossover probability $D$,
    we have
    \begin{align}
        I^\star(X;U\mid U')
        &= H^\star(X\mid U')-H^\star(X\mid U,U') \notag \\
        &= H(\theta)-H(D).
    \end{align}
    This matches the lower bound in~\eqref{eq:app-lb-2}, so
    $P^\star_{U,X,U'}$ is optimal.

    Finally, we compute the induced conditional distribution
    $P^{\star(\theta)}_{U\mid X,U'}$. By Bayes' rule,
    \begin{align}
        P^{\star(\theta)}_{U\mid X,U'}(u\mid x,u')
        &=
        \frac{
            P^\star_{X\mid U}(x\mid u)
            P^\star_{U\mid U'}(u\mid u')
        }{
            \sum_{\bar u\in\{0,1\}}
            P^\star_{X\mid U}(x\mid \bar u)
            P^\star_{U\mid U'}(\bar u\mid u')
        } .
    \end{align}
    Since
    \begin{equation}
        P^\star_{X\mid U}(x\mid u)
        =
        D\,2^{
            a\,\mathbf 1\{u=x\}
        },
        \qquad
        P^\star_{U\mid U'}(u\mid u')
        =
        \zeta\,2^{
            b\,\mathbf 1\{u=u'\}
        },
    \end{equation}
    where
    \begin{equation}
        a=\log\frac{1-D}{D},
        \qquad
        b=\log\frac{1-\zeta}{\zeta},
    \end{equation}
    the constants $D$ and $\zeta$ in the numerator and denominator cancel out.
    As a result,
    \begin{equation}
        P^{\star(\theta)}_{U \mid X,U'}(u\mid x,u')
        =
        \frac{
            2^{
                a\,\mathbf 1\{u=x\}
                +
                b\,\mathbf 1\{u=u'\}
            }
        }{
            \sum_{\bar u\in\{0,1\}}
            2^{
                a\,\mathbf 1\{\bar u=x\}
                +
                b\,\mathbf 1\{\bar u=u'\}
            }
        },
    \end{equation}
    which is exactly~\eqref{eq:optimizer}. Since
    $D,\zeta\!\in\!(0,1/2)$, we have $a,b\!>\!0$, and thus
    $P^{\star(\theta)}_{U\mid X,U'}\!\in\!\mc Q^\textup{NOST}$.
\end{proof}

\bibliographystyle{IEEEtran}
\bibliography{references}

\end{document}